\documentclass{osa-article}

\journal{oe}
\articletype{Research Article}
\usepackage{amsthm}
\usepackage[section]{placeins}
\usepackage{graphicx}
\usepackage{subcaption}
\usepackage{longtable}
\usepackage{lineno}
\newtheorem{theorem}{Theorem}[section]

\begin{document}

\title{Mie Scattering with 3D Angular Spectrum Method}

\author{Joel Lamberg\authormark{1,*}, Faezeh Zarrinkhat\authormark{1,2}, Aleksi Tamminen\authormark{1}, Mariangela Baggio\authormark{1}, Juha Ala-Laurinaho\authormark{1}, Juan Rius\authormark{2} Jordi Romeu\authormark{2}, Elsayed E. M. Khaled\authormark{3,4} and Zachary Taylor\authormark{1}}

\address{\authormark{1} Department of Electronics and Nanoengineering, MilliLab, Aalto University, Espoo, Finland\\
\authormark{2} CommSensLab, Technical University of Catalonia/UPC, Barcelona, Spain\\
\authormark{3} Department of Electrical Engineering, Assiut University, Assiut, Egypt\\
\authormark{4} High Institute of Engineering and Technology, Sohage, Egypt}

\email{\authormark{*} joel.lamberg@aalto.fi} %% email address is required



\begin{abstract}
Mie theory is a powerful method to model electromagnetic scattering from a multilayered sphere. Usually, the incident beam is expanded to its vector spherical harmonic representation defined by beam shape coefficients, and the multilayer sphere scattering is obtained by the T-matrix method. However, obtaining the beam shape coefficients for arbitrarily shaped incident beams has limitations on source locations and requires different methods when the incident beam is defined inside or outside the computational domain or at the scatterer surface. We propose a 3D angular spectrum method for defining beam shape coefficients from arbitrary source field distributions. This method enables the placement of the sources freely within the computational domain without singularities, allowing flexibility in beam design. We demonstrate incident field synthesis and spherical scattering by comparing morphology-dependent resonances to known values, achieving excellent matching and high accuracy. Additionally, we present mathematical proof to support our proposal. The proposed method has significant benefits for optical systems and inverse beam design. It allows for the analysis of electromagnetic forward/backward propagation between optical elements and spherical targets using a single method. It is also valuable for optical force beam design and analysis. 
\end{abstract}

%%%%%%%%%%%%%%%%%%%%%%%%%%  body  %%%%%%%%%%%%%%%%%%
\section{Introduction}
%%%%%%%%%%%%%%%%%%%%%%%%%%%%%%%%%%%%%%%%%%%%%%%%%%%%
Electromagnetic scattering from a homogeneous sphere illuminated by an arbitrary incident beam can be computed with conventional methods, including full-wave simulations \cite{Stratton1}, geometrical optics \cite{geo}, or physical optics \cite{PO,Sanford}. These techniques are well-studied and accurate, given the model fidelity and a suitable wavelength range. However, without considerable computational effort, they cannot assess the internal and scattered electric fields from multilayered spherical objects. Especially when the sphere's radius is of the order of wavelengths, complete classical electromagnetic wave theory is needed \cite{RanhaNeves:19}.

Mie theory and the generalized Lorentz-Mie theory are accurate methods to evaluate the internal and scattered fields from the multilayered spheres \cite{GOUESBET20111,GOUESBET20137,Wu:97}, where the incident beam is presented in vector spherical harmonics (VSH) expansion defined by beam shape coefficients (BSCs) \cite{bohrn}. However, obtaining the BSCs for arbitrarily shaped incident beams can be difficult; often needing a combination of several complex methods \cite{GOUESBET20111,Gouesbet:90}. The BSCs can be computed from a known function or field distribution with certain constraints, such as polarization, source shape, and location limitations \cite{ARB1,Gauss1}. In these cases, BSCs for the incident beam are obtained by the Bromwich method \cite{Maheu_1988} or multipole expansion \cite{Lock:06,Zvyagin:98,Taylor:09}. These methods define the sources inside a closed volume, and the VSH expanded fields can be computed only outside of this area, limiting the source location. 

BSCs can also be computed from known electric and magnetic field distributions on the surface of the spherical scatterer using closed surface orthogonality (CSO). CSO also has limitations as the fields must be defined across the entire spherical surface \cite{CSO}. For example, CSO cannot be applied to fields defined only on spherical subregion thus excluding it from many inverse beam synthesis tasks, where the beam only illuminates a small area of the sphere \cite{IRMMW2022}. On the other hand,  BSCs can be computed from the standard 2D ASM, which does not have these location restrictions due to its eigenfunction property. However, 2D ASM is limited to the planar surface distributions \cite{goodman,khaled}; see the comparison of the methods in Table (\ref{TABLE}).
\begin{table}[h!]
\begin{center}
\begin{tabular}{||m{9em}|m{3.3em}|m{3.3em}|m{3.3em}|m{3.3em}|m{3.3em}|m{3.3em}||} 
 \hline
Method &  Source outside comp. domain &  Source inside comp. domain &  Source on sphere´s sub-region &  \(RoC \sim\lambda\) & VSH expansion & Truly arb. field\\ 
 \hline\hline
Bromwich \cite{Maheu_1988} & \textcolor{purple}\checkmark &  &  & \textcolor{purple}\checkmark & \textcolor{purple}\checkmark & \textcolor{purple}\checkmark \\ 
 \hline
Multipole exp. \cite{Lock:06,Zvyagin:98,Taylor:09} &  \textcolor{purple}\checkmark &  &  & \textcolor{purple}\checkmark & \textcolor{purple}\checkmark & \textcolor{purple}\checkmark\\ 
 \hline
 CSO \cite{CSO} &  & \textcolor{purple}\checkmark &  & \textcolor{purple}\checkmark & \textcolor{purple}\checkmark & \textcolor{purple}\checkmark\\
  \hline
  Sanford et al. \cite{Sanford} & \textcolor{purple}\checkmark &  &  & \textcolor{purple}\checkmark & \textcolor{purple}\checkmark & \textcolor{purple}\checkmark\\
 \hline
  2D ASM \cite{goodman,khaled} & \textcolor{purple}\checkmark & \textcolor{purple}\checkmark & \textcolor{purple}\checkmark &  &  &\\
 \hline
  Stepwise ASM \cite{StepAS,Ebers:20}& \textcolor{purple}\checkmark & \textcolor{purple}\checkmark & \textcolor{purple}\checkmark &  &  &\\
 \hline
\textbf{\textcolor{teal}{3D ASM}} & \textcolor{teal}\checkmark & \textcolor{teal}\checkmark & \textcolor{teal}\checkmark & \textcolor{teal}\checkmark & \textcolor{teal}\checkmark & \textcolor{teal}\checkmark\\ 
 \hline
\end{tabular}
\end{center}
\caption{\small Characteristics of the selected methods in interest.}
\label{TABLE}
\end{table}\FloatBarrier
This article aims to create a VSH-expanded electromagnetic beam from a known electric field distribution on arbitrarily shaped and positioned surface. Then the incident beam could be modified in terms of propagating field distribution, polarization, and local phase variation. Also, the source location can be positioned inside the simulation area without restrictions. This goal can be achieved by expanding the 2D ASM to a 3D AS approach to accommodate arbitrary fields defined on arbitrarily shaped surfaces and construct modified BSCs for an incident beam VSH presentation.

2D ASMs present the incident field in the angular spectrum domain as the sum of differently oriented plane waves. When this angular summation of plane waves on a source plane is presented in a direction cosine coordinate system, these angles can be used directly to compute the BSCs for the VSH expansion. Additionally, the internal and scattered fields from a multilayered dielectric sphere can be mapped with the extended boundary condition method (EBCM), which is referred as the T-matrix method in this article\cite{khaled}. 

In \cite{StepAS}, the planar AS method was expanded to approximate diffractive fields from 2D curved surfaces. The expansion was obtained by dividing the curved 2D surface into the step-wise subregions windowed by the Gaussian function. This approach was later expanded to the 3D surfaces in \cite{Ebers:20}, where the authors used the same planar step-wise subsections with Gaussian distribution. Both methods approximate well the diffraction, reflections, and transmission from the curved single-layer boundaries when the radius of curvature (RoC) is much larger than the wavelength \cite{Worku:18}. 

The 3D ASM derivation begins with the same approach of dividing an arbitrary surface into planar subregions, presented as equal amplitude areas without Gaussian distribution or windowing (e.g. Stepwise AS). Then the areas of these subregions are shrunk into infinitesimally small points, which approach a Dirac delta function in the limit. This approach yields exact results (in the limit of infinitesimally small partitions) of the diffractive field from an arbitrary surface considering surface structures comparable to or larger than the illumination wavelength, i.e., \( k\alpha\geq2\pi\), where \( k\) is the wavenumber and \( \alpha\) is the radius of the sphere. 

Moreover, this result is analogous to the angular spectrum presentation of the Stratton-Chu method with magnetic dipole sources, where the electromagnetic radiation from a source point is calculated as a curl of a magnetic dipole multiplied by scalar Green's function \cite{6710963}. This presentation satisfies Maxwell's equations and approximates well curved surfaces when the RoC is larger than a wavelength \cite{Freni}. Also, the advantages of the presented method are non-singularity at the source points, allowing the synthesis of a continuous beam through the source surface from any known electric field distribution. Finally, the incident beam can be expanded into a VSH presentation to compute scattered fields from multilayered spheres.
%%%%%%%%%%%%%%%%%%%%%%%%%%  body  %%%%%%%%%%%%%%%%%%
%%%%%%%%%%%%%%%%%%%%%%%%%%%%%%%%%%%%%%%%%%%%%%%%%%%% 
\section{Theory}
%%%%%%%%%%%%%%%%%%%%%%%%%%%%%%%%%%%%%%%%%%%%%%%%%%%%
This chapter presents a method to compute the forward/backscatter of multilayered spheres under arbitrary beam illumination. The chapter is divided into three sections introducing the modified 3D ASM, global coordinate mapping system, and modified BSCs coefficients for the VSH expansion.

Section 2.1 begins by introducing a general parametrization of an arbitrary surface. The local base vectors are derived from parametrization to determine a Cartesian coordinate system for surface sources. Further, the electric field polarization is modified by base vector rotation. Also, a transformation matrix is presented for mapping the local source coordinates in a global coordinate system and vice versa. After coordinate mapping, an arbitrary surface parametrization is discretized into infinitesimally small surface elements. The key concept is that Riemann's surface integral combines differential surface elements, which can be approximated as locally planar elements. The ASM models the radiated field from each element, and the total electromagnetic field is the sum of the modeled fields from each source point by the superposition principle \cite{goodman}. This approach relies on the assumption that when the areas of the elements are small enough, the synthesized field from the differential elements approaches the field from the original surface. 

Section 2.2 presents a system to map local coordinates into a global origin-centered coordinate system with each source point's orientation and position information, which is later used to construct the BSCs.

Section 2.3 presents the VSH expansion from an arbitrary surface with modified BSCs. These BSCs are modified to include the mapped source's orientation and position information while preserving the needed spherical symmetry for VSH expansion.
%%%%%%%%%%%%%%%%%%%%%%%%%%%%%%%%%%%%%%%%%%%%%%%%%%%% 
\subsection{Modified angular spectrum method}
%%%%%%%%%%%%%%%%%%%%%%%%%%%%%%%%%%%%%%%%%%%%%%%%%%%%

Let the global Cartesian base-vectors  be (\(\mathbf{e}_x,\mathbf{e}_y,\mathbf{e}_z\)) in space \(\mathbb{R}^3\), where position vector is \(\mathbf{r}=x\mathbf{e}_x+y\mathbf{e}_y+z\mathbf{e}_z\) and global coordinates are marked as \(r_\text{glob}=(x,y,z)\). Consider a generalized, compact surface \(\Omega\) in \(\mathbb{R}^3\), which has a continuously differentiable parametrization with parameters \(p\) and \(q\) as
\begin{equation}\label{eq1}
\begin{array}{l}
\Omega = \Big\{\mathbf{o}(p,q)=o_x(p,q)\mathbf{e}_x+o_y(p,q)\mathbf{e}_y +o_z(p,q)\mathbf{e}_z \ | \\ p\in[p_1,p_2], \ q\in[q_1,q_2]\Big\}.
\end{array}
\end{equation}
Let us place the origin of the local coordinate system at point \(\mathbf{o}(p,q\)) on the surface \(\Omega\), where two base vectors \(\mathbf{f}_1\) and \(\mathbf{f}_2\) define tangential plane at the point, and the third base vector \(\mathbf{f}_3\) is a local surface normal vector, see Figure (\ref{surf}). When normalized, these vectors form an orthonormal base (\(\mathbf{f}_1,\mathbf{f}_2,\mathbf{f}_3\)) and are defined as

\begin{equation}\label{eq2}
\begin{array}{l}
\mathbf{f_3}=\big(\frac{\partial o}{\partial p}\times\frac{\partial o}{\partial q}\big)\big\|\frac{\partial o}{\partial p}\times\frac{\partial o}{\partial q}\big\|^{-1},\ \ \
\mathbf{f}_1=\big(\frac{\partial o}{\partial q}\big)\big\|\frac{\partial o}{\partial q}\big\|^{-1},\ \ \
\mathbf{f}_2=\mathbf{f}_3\times\mathbf{f}_1.
\end{array}
\end{equation}
The final local base (\(\mathbf{e}_1,\mathbf{e}_2,\mathbf{e}_3\)) is obtained from the base (\(\mathbf{f}_1,\mathbf{f}_2,\mathbf{f}_3\)), when the local polarization of the electric field is chosen.  Let us set
\begin{equation}\label{eq3}
\begin{array}{l}
\mathbf{e}_3 = \pm \mathbf{f}_3,
\end{array}
\end{equation}
where the sign is chosen according to the desired electromagnetic propagation direction. In a simple case, when polarization is known on the surface \(\Omega\), the local electric field vector \(\mathbf{e}_1\) on \(\Omega\) can be presented with the base (\(\mathbf{f}_1,\mathbf{f}_2\)), and the local magnetic field vector \(\mathbf{e}_2\) is obtained as a cross-product with local propagation direction \(\mathbf{e}_3\) as  \(\mathbf{e}_2=\mathbf{e}_3\times\mathbf{e}_1\). 

Polarization on the surface \(\Omega\) can also have externally sourced initial conditions, for example, from a polarization of an incident beam, that creates the electric field distribution on the surface \(\Omega\). Let \(\mathbf{p}=\mathbf{p}(p,q)\) be a vector that is perpendicular to the plane of that incident beam polarization. Then, local direction \(\mathbf{e}_1\) is an intersection of that plane, and the tangential plane of surface \(\Omega\) spanned by (\(\mathbf{f}_1,\mathbf{f}_2\)). 
\begin{figure}[!ht]
\centering
\includegraphics[width=0.55\linewidth]{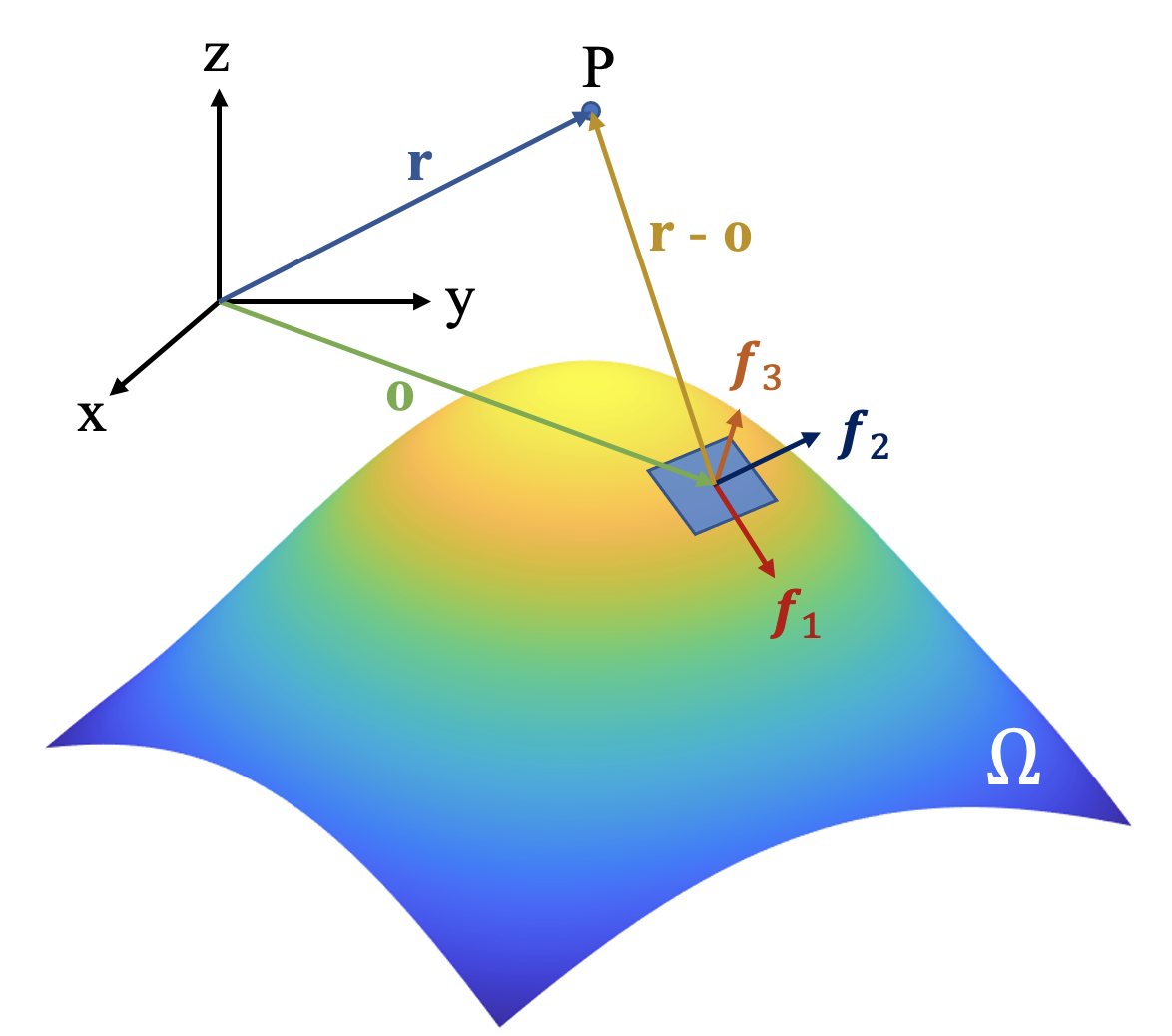}
\caption{\small Vector relations between source point base-vectors on an arbitrary surface and a global origin.}
\label{surf}
\end{figure}\FloatBarrier
Thus \(\mathbf{e}_1\), the base vector along the local electric field, and \(\mathbf{e}_2\), the base vector along the local magnetic field, are obtained as

\begin{equation}\label{eq4}
\begin{array}{l}
\mathbf{e}_1=(\mathbf{p}\times\mathbf{e}_3)\|\mathbf{p}\times\mathbf{e}_3\|^{-1}, \ \ \
\mathbf{e}_2=\mathbf{e}_3\times\mathbf{e}_1.
\end{array}
\end{equation}
The vectors \((\mathbf{e}_1,\mathbf{e}_2,\mathbf{e}_3\) form a right-handed orthonormal base, which together with the origin \(\mathbf{o}\) defines a local coordinate system. Let a vector \(\mathbf{r}\) be presented in this base as \(\mathbf{r}=\bar{x}\mathbf{e}_1+\bar{y}\mathbf{e}_2+\bar{z}\mathbf{e}_3\). Its coordinate triple is marked briefly as \(r_{\text{loc}}=(\bar{x},\bar{y},\bar{z})\). When the vector \(\mathbf{r}\) is expressed with both \((\mathbf{e}_x,\mathbf{e}_y,\mathbf{e}_z\)) and \((\mathbf{e}_1,\mathbf{e}_2,\mathbf{e}_3\)) bases, as it is known, the dependence of the coordinates on each other is determined by an orthogonal transformation matrix
\begin{equation}\label{eq5}
\begin{array}{l}
\Theta(p,q) = \left[\mathbf{e}_{1,{\text{glob}}}\,\,\mathbf{e}_{2,{\text{glob}}}\,\,\mathbf{e}_{3,{\text{glob}}}\right], \ \ \Theta^{-1}=\Theta^T,
\end{array}
\end{equation}
where the columns are the coordinate triples \(\mathbf{e}_{1,{\text{glob}}},\mathbf{e}_{2,{\text{glob}}}\) and \(\mathbf{e}_{3,{\text{glob}}}\).
Let the position \(P\) be presented in global coordinate system as \(\mathbf{r}=x\mathbf{e}_x+y\mathbf{e}_y+z\mathbf{e}_z\). Position \(P\) presented in the local coordinate system with a global base is
\begin{equation}\label{eq6}
\begin{array}{l}
\mathbf{r}-\mathbf{o}=(x-o_x)\mathbf{e}_x+(y-o_y)\mathbf{e}_y+(z-o_z)\mathbf{e}_z.
\end{array}
\end{equation}
Let´s denote briefly  \(\mathbf{\bar{r}}=\mathbf{r}-\mathbf{o}\). Thus, we have presentations of the location \(\mathbf{\bar{r}}\) in both local and global coordinate systems. The dependence of the corresponding coordinates can be written using the transformation matrix Eq. (\ref{eq5}) as
\begin{equation}\label{eq7}
\begin{array}{l}
(\mathbf{r}-\mathbf{o})_{\text{glob}}=\Theta\mathbf{\bar{r}}_{\text{loc}}.
\end{array}
\end{equation}

Let us divide the segment of the surface \(\Omega\) containing the electric field distribution \(E_0(p,q)\) into a finite partition of small separate sets \(\cup_c\Omega_c=\Omega\). Select the point  \(\mathbf{o}(p,q\)) on differential element \(\Omega_c\), the vectors associated with it are \(\mathbf{e}_1\) and \(\mathbf{e}_2\). Area \(\Omega_c\) is projected to the plane spanned by the vectors \(\mathbf{e}_1\) and \(\mathbf{e}_2\). The resulting area at the plane is marked as \(\Omega_t\). Likewise, function \(E_0\) is projected into that plane, limited to set \(\Omega_t\) and zero elsewhere. This geometry is presented in Figure \ref{fig:projection}. The resulting projection is marked as \(E_0^t\). In the local coordinate system it holds \(\bar{z}=0\) on the piece \(\Omega_t\), because this is located on the \(\bar{x}\bar{y}-\)plane. \(E_0^t\) is valid when \(\bar{x}\approx0\) and \(\bar{y}\approx0\). 
\begin{figure}[!ht]
\centering
\includegraphics[width=0.55\linewidth]{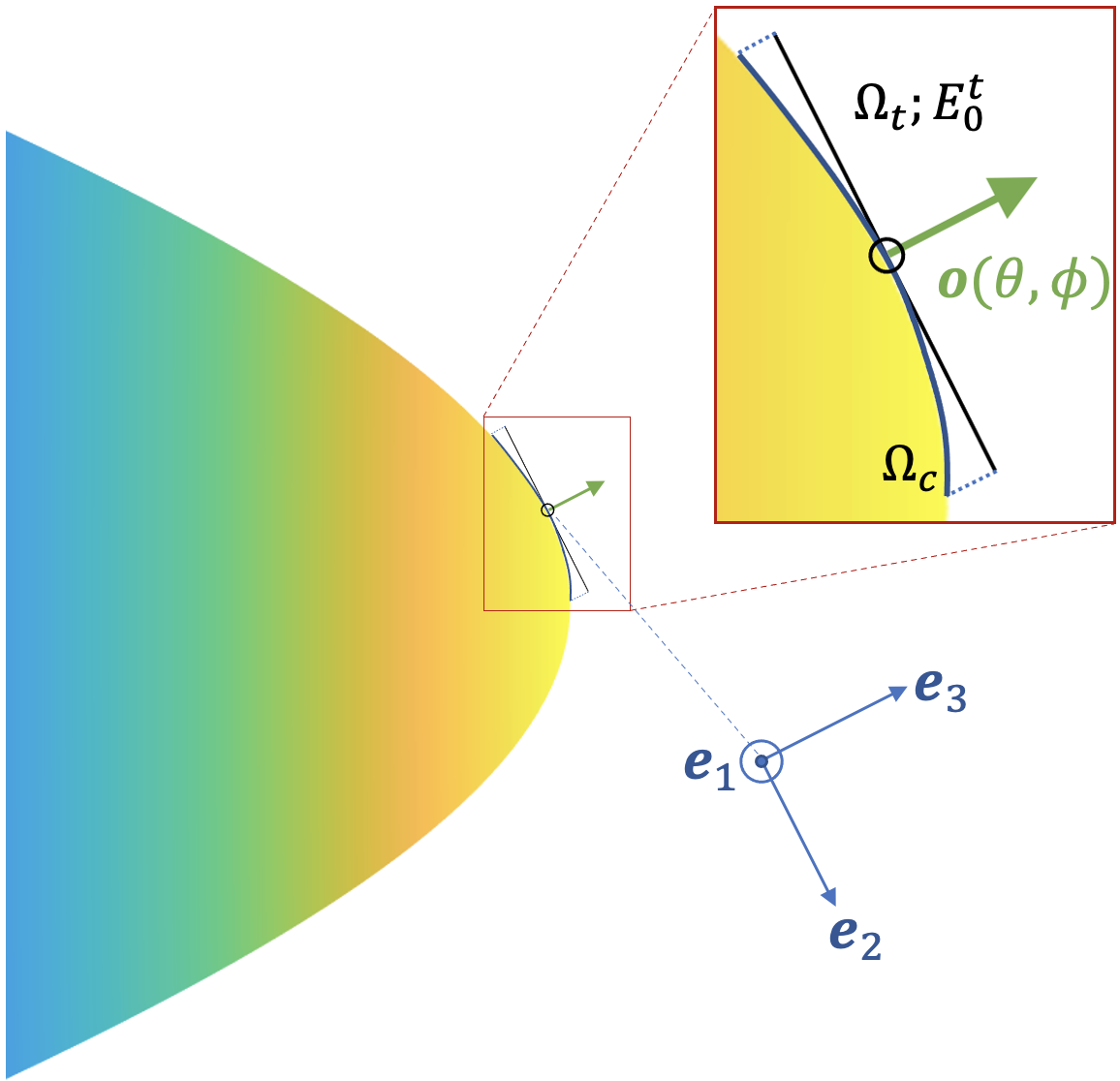}
\caption{\small Surface segment \(\Omega_c\) projected into a plane \(\Omega_t\).}
\label{fig:projection}
\end{figure}\FloatBarrier
The angular spectrum theory is locally applied on the local coordinate system to function \(E_0^t\) valid on piece \(\Omega_t\). Let's observe the field it creates at point \(\mathbf{\bar{r}}\), where the coordinates are as in Eq. (\ref{eq6}). At a fixed point \(\mathbf{r}\) holds \cite{goodman}
\begin{equation}\label{eq10}
\begin{array}{l}
E_t(\mathbf{r})=\frac{1}{4\pi^2}\iint_{\mathbb{R}^2}\mathcal{F}\big\{E_0^t\big\}(k_{\bar{x}},k_{\bar{y}})e^{i(k_{\bar{x}}\bar{x}+ k_{\bar{y}}\bar{y})}e^{i|\bar{z}|\sqrt{k^2-k_{\bar{x}}^2-k_{\bar{y}}^2}}dk_{\bar{x}}dk_{\bar{y}},
\end{array}
\end{equation}
where \(\mathcal{F}\big\{E_0^t\big\}\) is the Fourier transform of the local electric field \(E_0^t\) and \(k=(k_x^2+k_y^2+k_z^2)^{1/2}\) is the wavenumber. The total field from the arbitrary surface is obtained by summing the fields from the differential sources by the superposition principle as
\begin{equation}\label{eq11}
\begin{array}{l}
\begin{aligned}
\mathbf{E}_1(\mathbf{r})&\approx\sum_t E_t(\mathbf{r})\mathbf{e}_1=\sum_t\mathbf{e}_1\frac{1}{4\pi^2}\iint_{\mathbb{R}^2}\mathcal{F}\big\{E_0^t\big\}(k_{\bar{x}},k_{\bar{y}})e^{i(k_{\bar{x}}\bar{x}+k_{\bar{y}}\bar{y})}\\&\times e^{i|\bar{z}|\sqrt{k^2-k_{\bar{x}}^2-k_{\bar{y}}^2}}dk_{\bar{x}}dk_{\bar{y}},
\end{aligned}
\end{array}
\end{equation}
where local coordinates \((\bar{x},\bar{y},\bar{z})\) as in Eq. (\ref{eq6}) and the total field is polarized along the global \(\mathbf{e}_1\) unit-vector. The obtained field Eq. (\ref{eq11}) also satisfies the Helmholtz equation as unit vector \(\mathbf{e}_1\) is locally constant. To be precise, if the limit with respect to the areas \(|\Omega_t|\) of the sets \(\Omega_t\) exists, we define
\begin{equation}\label{eq12}
\begin{array}{l}
\mathbf{E}_1(\mathbf{r})=\mathbf{E}_1(x,y,z)=lim_{|\Omega_t|\to 0} \sum_t E_t(\mathbf{r})\,\mathbf{e}_1.
\end{array}
\end{equation}
As pointed out at the beginning of the Theory section, this limit also works physically when considering smooth surfaces and is even more precise when the wavelength decreases compared to the local radius of the surface. Finally, a further examination is made to obtain a closed form for clause Eq. (\ref{eq12}); we present a proof in Appendix A.
 
Heuristically, when \(|\Omega_t|\) shrinks, the function \(E_0^t/[E_0^t(\bar{0},\bar{0})|\Omega_t|]\), normed in volume, approaches a Dirac delta \(\delta_{(0,0)}\) in the sense of the distribution theory; \(E_0^t/[E_0^t(\bar{0},\bar{0})|\Omega_t|]\rightarrow\delta_{(0,0)}\) thus \(E_0^t\rightarrow E_0^t(\bar{0},\bar{0})|\Omega_t|\delta_{(0,0)}\). As known, the Fourier transform of Dirac delta is 1, thus, by small \(|\Omega_t|\)
\begin{equation}\label{eq13}
\begin{array}{l}
\begin{aligned}
\mathcal{F}\big\{E_0^t\big\}(k_{\bar{x}},k_{\bar{y}})&=\iint_{\mathbb{R}^2} E_0^t e^{-i(k_{\bar{x}}\bar{x}+ k_{\bar{y}}\bar{y})}d\bar{x}d\bar{y}\\
&\approx\iint_{\mathbb{R}^2} E_0^t(\bar{0},\bar{0})|\Omega_t|\delta_{(0,0)} e^{-i(k_{\bar{x}}\bar{x}+ k_{\bar{y}}\bar{y})}d\bar{x}d\bar{y}\\
&=E_0(\theta,\phi)|\Omega|\mathcal{F}\big\{\delta_{(0,0)}\big\}=E_0(\theta,\phi)|\Omega_t|,
\end{aligned}
\end{array}
\end{equation}
since \(E_1(p,q)=E_1^t(\bar{0},\bar{0})\) when \(\mathbf{o}\) is fixed. Then Eq. (\ref{eq11}) and Eq. (\ref{eq13}) yield first the form of a Riemann sum and then an integral as a limit
\begin{equation}\label{eq14}
\begin{array}{l}
\begin{aligned}
\mathbf{E}_1(\mathbf{r})&=\mathbf{E}_1(x,y,z)=lim_{|\Omega_t|\to 0} \sum_t\mathbf{e}_1\frac{1}{4\pi^2} E_0(p,q)|\Omega_t|\\&\times\iint_{\mathbb{R}^2}e^{i(k_{\bar{x}}\bar{x}+ k_{\bar{y}}\bar{y})}e^{i|\bar{z}|\sqrt{k^2-k_{\bar{x}}^2-k_{\bar{y}}^2}}dk_{\bar{x}}dk_{\bar{y}}\\&=\frac{1}{4\pi^2}\iint_{S}E_0(p,q)\mathbf{e}_1(\theta,\phi)\iint_{\mathbb{R}^2}e^{i(k_{\bar{x}}\bar{x}+k_{\bar{y}}\bar{y})}e^{i|\bar{z}|\sqrt{k^2-k_{\bar{x}}^2-k_{\bar{y}}^2}}dk_{\bar{x}}dk_{\bar{y}}d\Omega.
\end{aligned}
\end{array}
\end{equation}
The final form can be written as
\begin{equation}\label{eq15}
\begin{array}{l}
\mathbf{E}_1(\mathbf{r})=\frac{1}{4\pi^2}\iint_{\Omega}E_0(p,q)\mathbf{E}_t(\mathbf{r};\theta,\phi)\|\frac{\partial O}{\partial p}\times\frac{\partial O}{\partial q}\|d\theta d\phi,
\end{array}
\end{equation}
where
\begin{equation}\label{eq16}
\begin{array}{l}
\mathbf{E}_t(\mathbf{r};p,q)=\mathbf{e}_1(p,q)\iint_{\mathbb{R}^2}e^{i(k_{\bar{x}}\bar{x}+k_{\bar{y}}\bar{y})}e^{i|\bar{z}|\sqrt{k^2-k_{\bar{x}}^2-k_{\bar{y}}^2}}dk_{\bar{x}}dk_{\bar{y}}\\
\end{array}
\end{equation}
and \(\bar{x},\bar{y}\) and \(\bar{z}\) are as in Eq. (\ref{eq6}) and notation \((\mathbf{r};p,q)\) defines electric field on observation point \(\mathbf{r}\) related to the location (\(p,q)\) on the parametrizied surface. Propagating waves are obtained by integrating Eq. (\ref{eq16}) over the real disk \(k_{\bar{x}}^2+k_{\bar{y}}^2\leq k^2\), and evanescence waves can be obtained by expanding the integration over the real domain, i.e., \(k_{\bar{x}}^2+k_{\bar{y}}^2> k^2\).

There are two key observations regarding this derivation. First, Eq. (\ref{eq16}) is similar to \(\mathbf{E}_t=-\nabla\times(\mathbf{M}G)\) after replacing the scalar Green's function (\(G\)) by its angular spectrum given by Weyl identity, where \(\mathbf{M}=-2\mathbf{n}\times\mathbf{e}_2\) is a small magnetic dipole \cite{Chew}. Thus, Eq. (\ref{eq15}) is a particular case of the Stratton-Chu equation for the electric field in which only the magnetic sources \(\mathbf{M}=-\mathbf{n}\times\mathbf{E}_{surf}\) are presented and the electric ones \(\mathbf{J}=\mathbf{n}\times\mathbf{H}_{surf}\) have been removed \cite{Stratton}. Secondly, if the surface \(\Omega\) is planar, Eq. (\ref{eq15}) and Eq. (\ref{eq16}) return to the original two-dimensional angular spectrum method.

%%%%%%%%%%%%%%%%%%%%%%%%%%%%%%%%%%%%%%%%%%%%%%%%%%%%
\subsection{Polarization and global coordinate mapping}
%%%%%%%%%%%%%%%%%%%%%%%%%%%%%%%%%%%%%%%%%%%%%%%%%%%%
An electromagnetic beam in the plane wave spectrum representation also has polarization vector components along the propagation direction. To account for polarization let us expand Eq. (\ref{eq16}) along plane \((\mathbf{e}_1,\mathbf{e}_3)\) using the relation \(\mathbf{E}_t\cdot\mathbf{k}=0\) \cite{esam94}. Let \(\mathbf{k}=k_{\bar{x}}+k_{\bar{y}}+k_{\bar{z}}\) be the wave vector, where \(k=\|\mathbf{k}\|\) is the wave number and \(k_z=[k^2-k_{\bar{x}}^2-k_{\bar{y}}^2]^{1/2}\). Moreover let´s denote \(\mathbf{s}=\frac{1}{k}\mathbf{k}\) and define \(\mathbf{r}_{||}=\bar{x}\mathbf{e}_1+\bar{y}\mathbf{e}_2+|\bar{z}|\mathbf{e}_3\). Hence the summarized exponent in Eq. (\ref{eq16}) can be written as \(i(k_{\bar{x}}\bar{x}+k_{\bar{y}}\bar{y}+k_{\bar{z}}|\bar{z}|)=i\mathbf{k}\cdot\mathbf{r}_{||}=ik(\mathbf{s}\cdot\mathbf{r}_{||})\), obtaining
\begin{equation}\label{eq19}
\begin{array}{l}
\mathbf{E}_t(\mathbf{r};p,q)=\iint_{\mathbb{R}^2}e^{ik(\mathbf{\bar{s}}\cdot\mathbf{\bar{r}}_{||})}\Big[\mathbf{e}_1-\Big(\frac{k_{\bar{x}}}{k_{\bar{z}}}\Big)\mathbf{e}_3\Big]d k_{\bar{x}}d k_{\bar{y}}=\iint_{\mathbb{R}^2}e^{ik(\mathbf{\bar{s}}\cdot\mathbf{\bar{r}}_{||})}\mathbf{e}_0(k_{\bar{x}},k_{\bar{y}},k_{\bar{z}})d k_{\bar{x}}d k_{\bar{y}},
\end{array}
\end{equation}
where \(\mathbf{e}_0(k_{\bar{x}},k_{\bar{y}},k_{\bar{z}})\) defines the polarization. The local Cartesian coordinate system presents differential source Eq. (\ref{eq19}) in the base \((\mathbf{e}_1,\mathbf{e}_2,\mathbf{e}_3\)). When expanded to the VSH presentation, the influence of each differential source must be presented by means of position vector \(\mathbf{r}\) to preserve the spherical symmetry of VSH. Next, the global mapping of local coordinates is presented.

We have a simple connection \((\mathbf{r}-\mathbf{o})_{\text{glob}}=\Theta\,\mathbf{\bar{r}}_{\text{loc}}\) as in Eq. (\ref{eq6}). Let us define \(\mathbf{s}_{||}=\mathbf{s}\), when \(\bar{z}>0\), but, when \(\bar{z}<0\), the sign of \(\bar{z}-\)coordinates in \(\mathbf{s}\) is exchanged. Consequently, we can write 
\begin{equation}\label{eqSS}
\begin{array}{l}
\mathbf{\bar{s}}\cdot\mathbf{\bar{r}}_{||}=\mathbf{\bar{s}}_{||}\cdot\mathbf{\bar{r}}=\mathbf{\bar{s}}_{{||},\text{loc}}\cdot\mathbf{\bar{r}}_\text{loc}=\mathbf{t}\cdot\mathbf{\bar{r}}_\text{loc},
\end{array}
\end{equation}
where written briefly \(\mathbf{t}=\mathbf{\bar{s}}_{{||},{loc}}\). As a result, we get
\begin{equation}\label{eq25}
\begin{array}{l}
\mathbf{\bar{s}}\cdot\mathbf{\bar{r}}_{||}=\mathbf{\bar{s}}_{||}\cdot\mathbf{\bar{r}}=\mathbf{t}\cdot\mathbf{\bar{r}}_\text{loc}=\mathbf{t}\cdot\Theta^{T}(\mathbf{r}-\mathbf{o})_{\text{glob}}=\\ \mathbf{t}\cdot(\Theta^T \mathbf{r}_\text{glob}-\Theta^T \mathbf{o}_\text{glob})=\mathbf{t}\cdot(\mathbf{r}_\text{loc}-\mathbf{o}_\text{loc})=\mathbf{t}\cdot\mathbf{r}_\text{loc}-\mathbf{t}\cdot\mathbf{o}_\text{loc}.
\end{array}
\end{equation}
Based on Eq. (\ref{eq25}), the term \(e^{ik(\mathbf{s}\cdot\mathbf{\bar{r}}_{||})}\), which is in the local coordinate system, can be presented in the global coordinate system as

\begin{equation}\label{eq26}
\begin{array}{l}
e^{ik(\mathbf{s}\cdot\mathbf{\bar{r}}_{||})}=e^{ik(\mathbf{t}\cdot\mathbf{r}_{\text{loc}})}e^{-ik(\mathbf{t}\cdot\mathbf{o}_{\text{loc}})},
\end{array}
\end{equation}
where \(e^{-ik(\mathbf{t}\cdot\mathbf{o}_{\text{loc}})}=e^{-ik(\mathbf{s}_{||}\cdot\mathbf{o})}\) presents the position and phase shift of the source beam's focus (waist) from global origin to the source surface. Each local beam created from the differential source is presented in a global coordinate system.

%%%%%%%%%%%%%%%%%%%%%%%%%%%%%%%%%%%%%%%%%%%%%%%%%%%%
\subsection{VSH expansion with modified beam shape coefficients}
%%%%%%%%%%%%%%%%%%%%%%%%%%%%%%%%%%%%%%%%%%%%%%%%%%%%

Modified BSCs coefficients for each source point are derived for presenting any parameterized field distribution in a global coordinate system, including the locations and orientations of the sources. These coefficients map complicated coordinate geometries into one presentation. This allows synthesizing a total electromagnetic beam as a superposition of VSH expanded fields.

The integration wave numbers \(k_{\bar{x}},k_{\bar{y}}\) in Eq. (\ref{eq19}) are analogous to the BSC angles \cite{khaled}. Thus Eq. (\ref{eq19}) is  expressed in the direction cosine coordinate system as

\begin{equation}\label{eq20}
\begin{array}{l}
\mathbf{E}_t(\mathbf{r};p,q)=k^2\int_{0}^{\pi}\sin{\xi}\Big\{\int_{0}^{\pi}\mathbf{e}_0(\xi,\zeta)e^{ik(\mathbf{t}\cdot\mathbf{r}_{\text{loc}})}e^{-ik(\mathbf{t}\cdot\mathbf{o}_{\text{loc}})}\sin{\zeta}d\zeta\Big\} d\xi,
\end{array}
\end{equation}
where \(\xi\) and \(\zeta\) are direction cosines related to Cartesian coordinates as \(k\cos{\xi}=k_{\bar{x}}\) and \(k\cos{\zeta}=k_{\bar{y}}\). The polarization \(\mathbf{e}_0\) in terms of direction cosines is given as
\begin{equation}\label{eq21}
\begin{array}{l}
\mathbf{e}_0(\xi,\zeta)=\mathbf{e}_1-\frac{\cos{\xi}}{(1-\cos^2{\xi}+\cos^2{\zeta})^{1/2}}\mathbf{e}_3,
\end{array}
\end{equation}
and the vector \(\mathbf{s}\) is
\begin{equation}\label{eq22}
\begin{array}{l}
\mathbf{s}=\cos{\xi}\mathbf{e}_1+\cos{\zeta}\mathbf{e}_2+(1-\cos^2{\xi}+\cos^2{\zeta})^{1/2}\mathbf{e}_3.
\end{array}
\end{equation}
Then, Eq. (\ref{eq20}) is evaluated by numerical integration over a disk of radius \(k\) using the trapezoidal rule method, which is sufficiently accurate for periodic functions \cite{trap}, with uniform-width \(l\) \cite{khaled}
\begin{equation}\label{eq23}
\begin{array}{l}
\mathbf{E}_t\approx p^2\sum_i\sin{\xi_i}\sum_j\sin{\zeta_j}\mathbf{e}_{0_{ij}}(\xi_{t_i},\zeta_{t_j})e^{ik(\mathbf{t}_{ij}\cdot\mathbf{r}_{\text{loc}})}e^{-ik(\mathbf{t}_{ij}\cdot\mathbf{o}_{\text{loc}})},
\end{array}
\end{equation}
where indices \(i\) and \(j\) refer to individual plane waves with propagation constants \(k_{\bar{x}}^i\) and \(k_{\bar{y}}^j\) respectively to \(\bar{x}\) and \(\bar{y}\) directions, and \(p=\xi_{t+1}-\xi_t=\zeta_{t+1}-\zeta_t\)
is the step size of the numerical integration and \(\mathbf{t}_{ij}=(\mathbf{s}_{{||},{ij}})_\text{loc}\). The integration radius is limited to \(k\leq 2\pi/\lambda\) for obtaining only propagating waves, and the evanescence fields can be obtained by expanding the integration domain. 

The transformation matrix from local coordinates (\(\bar{x},\bar{y},\bar{z}\)) to global coordinates (\(x,y,z\)) is \(\Theta(p,q)\) as in Eq. (\ref{eq5}), 
 and the radiated electric field can be written at an arbitrary point \(\mathbf{r}\)  by VSHs  as \cite{khaled}
\begin{equation}\label{eq27}
\begin{array}{l}
\begin{aligned}
\mathbf{E}_t(\mathbf{r})&=p^2\sum_m\sum_n D_{mn}\Big[a^t_{\text{emn}}\mathbf{M}^1_{\text{emn}}(k\mathbf{r})+a^t_{\text{omn}}\mathbf{M}^1_{\text{omn}}(k\mathbf{r})\\&+b^t_{\text{emn}}\mathbf{N}^1_{\text{emn}}(k\mathbf{r})+b^t_{\text{omn}}\mathbf{N}^1_{\text{omn}}(k\mathbf{r})\Big],
\end{aligned}
\end{array}
\end{equation}
where \(\mathbf{M}^1_{\text{emn}},\mathbf{M}^1_{\text{omn}},\mathbf{N}^1_{\text{emn}}\) and \(\mathbf{N}^1_{\text{omn}}\) are the VSH of the 
first kind and \(D_{mn}\) is a normalization factor \cite{khaled}. For \(k\mathbf{r}\) we have to substitute its spherical coordinates \( (kr,\theta,\phi) \) and corresponding spherical base vectors \( (\mathbf{e}_\theta,\mathbf{e}_\phi) \), taken with respect to the local base \( (\mathbf{e}_1,\mathbf{e}_2,\mathbf{e}_3) \). They can be computed from the triple \(r_\text{loc}\). Modified BSCs \(a^t_{\text{emn}},a^t_{\text{omn}},b^t_{\text{emn}}\) and \(b^t_{\text{omn}}\) are  given as

\begin{equation}\label{eq29}
\begin{array}{l}
a^t_{\text{emn}}=\sum_i\sin{\xi_i}\sum_j\sin{\zeta_j}a^t_{{\text{emn}}_{ij}}e^{-ik(\mathbf{t}_{ij}\cdot\mathbf{o}_\text{loc})},\\
a^t_{\text{omn}}=\sum_i\sin{\xi_i}\sum_j\sin{\zeta_j}a^t_{{\text{omn}}_{ij}}e^{-ik(\mathbf{t}_{ij}\cdot\mathbf{o}_\text{loc})},\\
b^t_{\text{emn}}=\sum_i\sin{\xi_i}\sum_j\sin{\zeta_j}b^t_{{\text{emn}}_{ij}}e^{-ik(\mathbf{t}_{ij}\cdot\mathbf{o}_\text{loc})},\\
b^t_{\text{omn}}=\sum_i\sin{\xi_i}\sum_j\sin{\zeta_j}b^t_{{\text{omn}}_{ij}}e^{-ik(\mathbf{t}_{ij}\cdot\mathbf{o}_\text{loc})}.
\end{array}
\end{equation}
The BSCs for each \(ij-\) plane wave \(a^t_{{\text{emn}}_{ij}}, a^t_{{\text{omn}}_{ij}}, b^t_{{\text{emn}}_{ij}}\) and \(b^t_{{\text{omn}}_{ij}}\) are defined in Appendix B. Now the magnetic field is obtained by rearranging the modified BSCs and by multiplying by constant \(-i/\eta\) as
\begin{equation}\label{EQ27}
\begin{array}{l}
\begin{aligned}
\mathbf{H}_t(\mathbf{r})&=-\frac{ip^2}{\eta}\sum_m\sum_n D_{mn}\Big[b^t_{\text{emn}}\mathbf{M}^1_{\text{emn}}(k\mathbf{r})+b^t_{\text{omn}}\mathbf{M}^1_{\text{omn}}(k\mathbf{r})\\&+a^t_{\text{emn}}\mathbf{N}^1_{\text{emn}}(k\mathbf{r})+a^t_{\text{omn}}\mathbf{N}^1_{\text{omn}}(k\mathbf{r})\Big],
\end{aligned}
\end{array}
\end{equation}
The total incident electric field Eq. (\ref{eq15}) and total incident magnetic field from the source distribution, accounting for polarization as in Eq. (\ref{eq19}), is written as the integration of vector spherical harmonics expansions Eq. (\ref{eq27}) and Eq. (\ref{EQ27}) over the source surface as 
\begin{equation}\label{INC}
\begin{array}{l}
\mathbf{E}(\mathbf{r})_{inc}=\frac{1}{4\pi^2}\iint_{\Omega}E_0(p,q)\mathbf{E}_t(\mathbf{r};p,q)\|\frac{\partial o}{\partial p}\times\frac{\partial o}{\partial q}\|dpdq,\\
\mathbf{H}(\mathbf{r})_{inc}=\frac{1}{4\pi^2}\iint_{\Omega}E_0(p,q)\mathbf{H}_t(\mathbf{r};p,q)\|\frac{\partial o}{\partial p}\times\frac{\partial o}{\partial q}\|dpdq.
\end{array}
\end{equation}
Similar scattered field presentations mapped with the T-matrix method are presented as

\begin{equation}\label{SCAT}
\begin{array}{l}
\mathbf{E}(\mathbf{r})_{sca}=\frac{1}{4\pi^2}\iint_{\Omega}E_0(p,q)\mathbf{E}_{scat}(\mathbf{r};p,q)\|\frac{\partial o}{\partial p}\times\frac{\partial o}{\partial q}\|dpdq, \\
\mathbf{H}(\mathbf{r})_{sca}=\frac{1}{4\pi^2}\iint_{\Omega}E_0(p,q)\mathbf{H}_{scat}(\mathbf{r};p,q)\|\frac{\partial o}{\partial p}\times\frac{\partial o}{\partial q}\|dpdq
\end{array}
\end{equation}
where \(\mathbf{E}_{scat}\) and \(\mathbf{H}_{scat}\) are introduced in Appendix B.

In summary, we have expanded planar 2D ASM to 3D ASM by dividing the surface into infinitesimally small sub-regions used as local source points. We mapped source points' global location and orientation information into the modified BSCs and obtained the total fields as a superposition of VSH-expanded source points. Also, surface polarization can be modified by parameterization.
%%%%%%%%%%%%%%%%%%%%%%%%%%  body  %%%%%%%%%%%%%%%%%%
\section{Results}
The results are divided into three sections; first, we present the radiation patterns and non-singularity of the differential sources. In the second section, we demonstrate the practicality and accuracy of the presented theory by comparing the Gaussian beam´s morphology-dependent resonances from a sphere with the traditional 2D ASM and the presented 3D ASM, with a high agreement. In the last section, we synthesize an electromagnetic beam from an ellipsoidal surface, verify the incident field with physical optics simulations, and compute scattered fields from a 100-dielectric layer sphere.
%%%%%%%%%%%%%%%%%%%%%%%%%%%%%%%%%%%%%%%%%%%%%%%%%%%%%
%%%%%%% Differential source as a magnetic dipole %%%%
%%%%%%%%%%%%%%%%%%%%%%%%%%%%%%%%%%%%%%%%%%%%%%%%%%%%%
\subsection{Differential source}
The total field from arbitrary surface \(\Omega\) is obtained as a superposition of differential source fields \(\mathbf{E}_{tot}=\sum_{\Omega}\mathbf{E}_t\). As shown in derivation, when the surface area of the differential source plane approaches zero, the result Eq. (\ref{eq16}) is analagous to the angular spectrum presentation of a curl of the magnetic dipole multiplied with scalar Green's function \(\mathbf{E}_t=-\nabla\times(-2\mathbf{n}\times\mathbf{e}_2G)\). As is well known, the fields created in this way are exact from an infinitely large plane, and the method accurately approximates field synthesis from curved surfaces when the RoC is above a wavelength.

The radiation pattern of a differential source \(\mathbf{E}_t\) is presented,  where the magnetic dipole is positioned to the origin along the y-axis. Radiated fields are computed at transverse xy-plane and yz-plane; see Figure (\ref{dipole}). This arrangement's main polarization is along the x-axis, and the field propagates to the positive z-direction. 
\begin{figure}[!ht]
\centering
\includegraphics[width=1\linewidth]{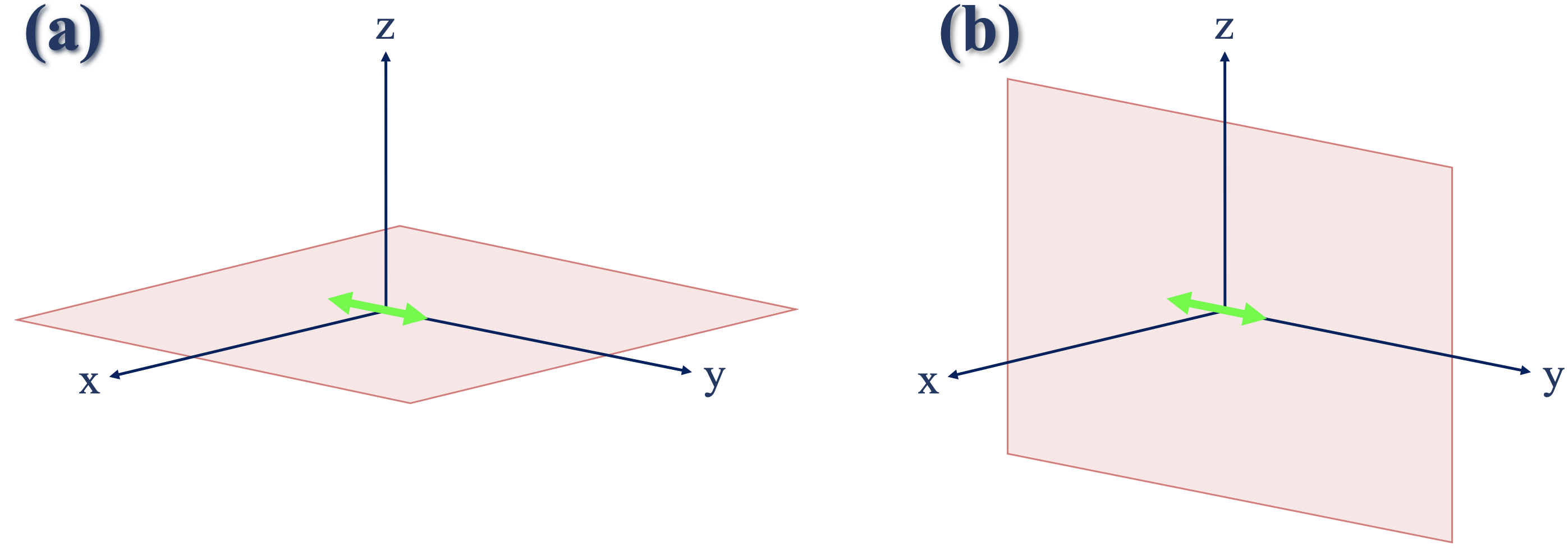}
\caption{\small Magnetic dipole in \(\mathbf{E}_t\) along the y-axis. The magnetic dipole is marked as a green arrow on the evaluation plane presented as a red area. On a) and b) are the simulation arrangements on xy- and yz-planes, respectively. }
\label{dipole}
\end{figure}\FloatBarrier
Figure (\ref{DiffSource}) presents the radiated fields from a differential source on xy- and yz-planes. 
\begin{figure}[!h]
\centering
\includegraphics[width=1\linewidth]{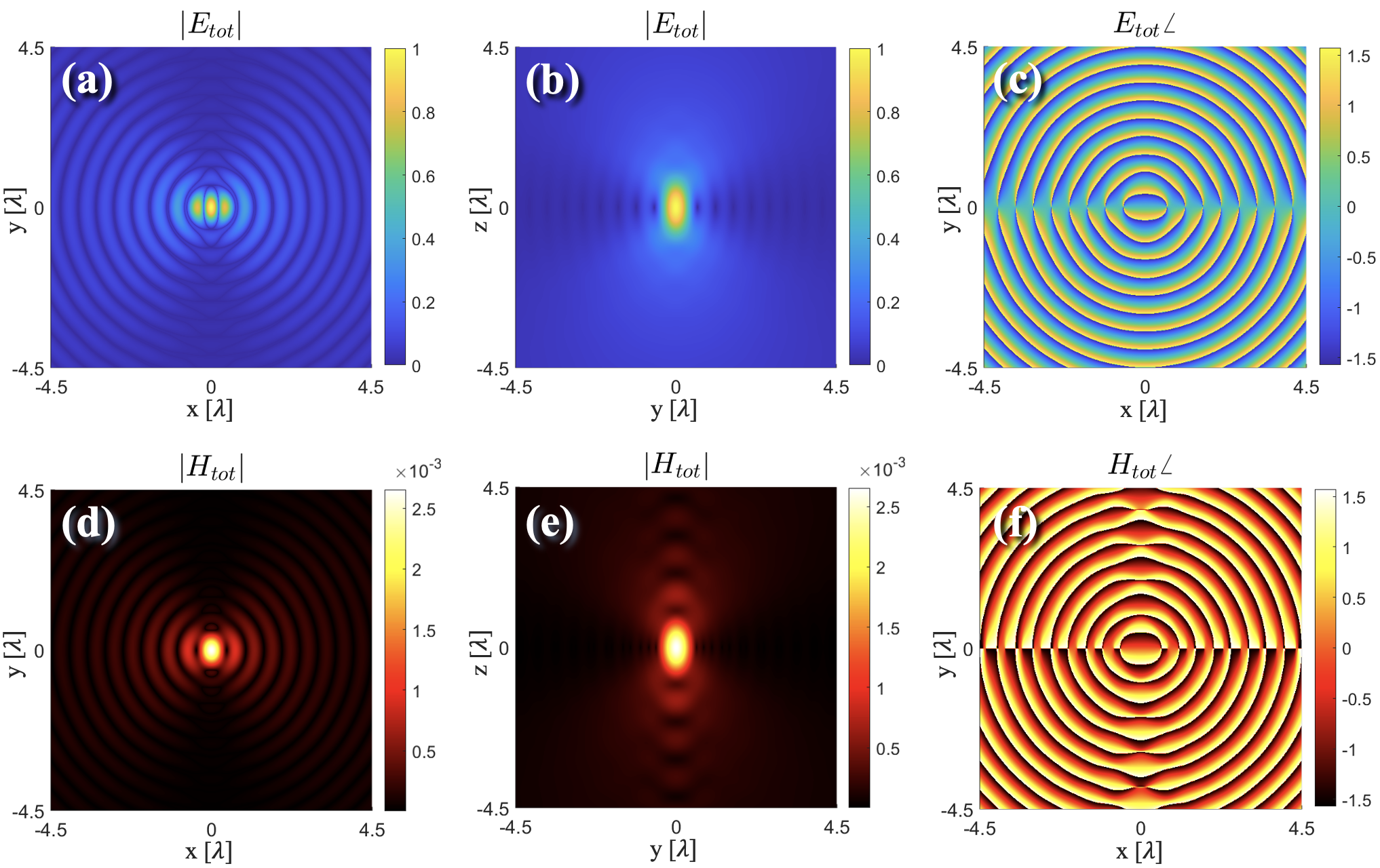}
\caption{\small Electromagnetic radiation from a source point, where a) electric and d) magnetic fields at the transverse xy-plane. On b) electric and e) magnetic fields on propagation yz-plane, and on c) electric and f) magnetic fields phases on propagation yz-plane.}
\label{DiffSource}
\end{figure}\FloatBarrier
Simulation shows the free placement of source points inside the computation domain due to the absence of singularities. The maximum amplitude of the electric fields is one without normalization.
%%%%%%%%%%%%%%%%%%%%%%%%%%%%%%%%%%%%%%%%%%%%%%%%%%%%%
%%%%%%% Morphology-dependent resonance test %%%%%%%%%
%%%%%%%%%%%%%%%%%%%%%%%%%%%%%%%%%%%%%%%%%%%%%%%%%%%%%
\subsection{Morphology-dependent resonance test}
Mie scattering region is known for frequency-dependent backscatter intensity. This phenomenon is due to the scattering resonances, which trap energy temporally inside the sphere. Backscatter intensity is presented as a function of size parameter \(k\alpha\), where \(\k=k_0n\) is a refractive index dependent wavenumber inside the sphere and \(\alpha\) is the radius of the sphere. These backscatter resonances are called morphology-dependent resonances (MDRs) \cite{MDR1,MDR2}. MDRs are highly responsive to simulation errors and are a good test of the accuracy of the derived 3D ASM combined with VSH expansion \cite{MDR3}. 

High accuracy manifests with simulations, where a Gaussian beam is created from its beam waist and its field components on a plane outside the beam waist. In the latter case, each source point has varying propagation direction and complex amplitude, making field synthesis on this computation inefficient with traditional methods. MDRs are computed with both source plane locations when the beam waist is positioned at the origin of the homogeneous sphere and on the side of the coated sphere, defined precisely later. Likewise, MDRs from an individual source point are computed in both focus scenarios, see Figure (\ref{GaussianMr136}).

Simulations are executed by illuminating homogeneous and coated spheres by Gaussian beam in size parameter range \(k\alpha\in[32,36]\) with beam waist radius \(\omega_0=1.5\lambda\). Homogeneous spheres have a refractive index of \(n_r=1.36\), and coated spheres have a \(n_r=1.36\) core with a \(n_c=1.5\) shell with radius of \(0.7\ \alpha\). These values were selected from article \cite{esam94} for comparing results. In both scenarios, the Gaussian beam is first synthesized on the transverse Plane 1 at the beam waist from the planar Gaussian electric field distribution with the traditional 2D AS and the proposed 3D AS method with an exact match. Then the incident field is computed on Plane 2 at \(-5 \ \alpha\) distance from the origin. 
\begin{figure}[!h]
\centering
\includegraphics[width=0.8\linewidth]{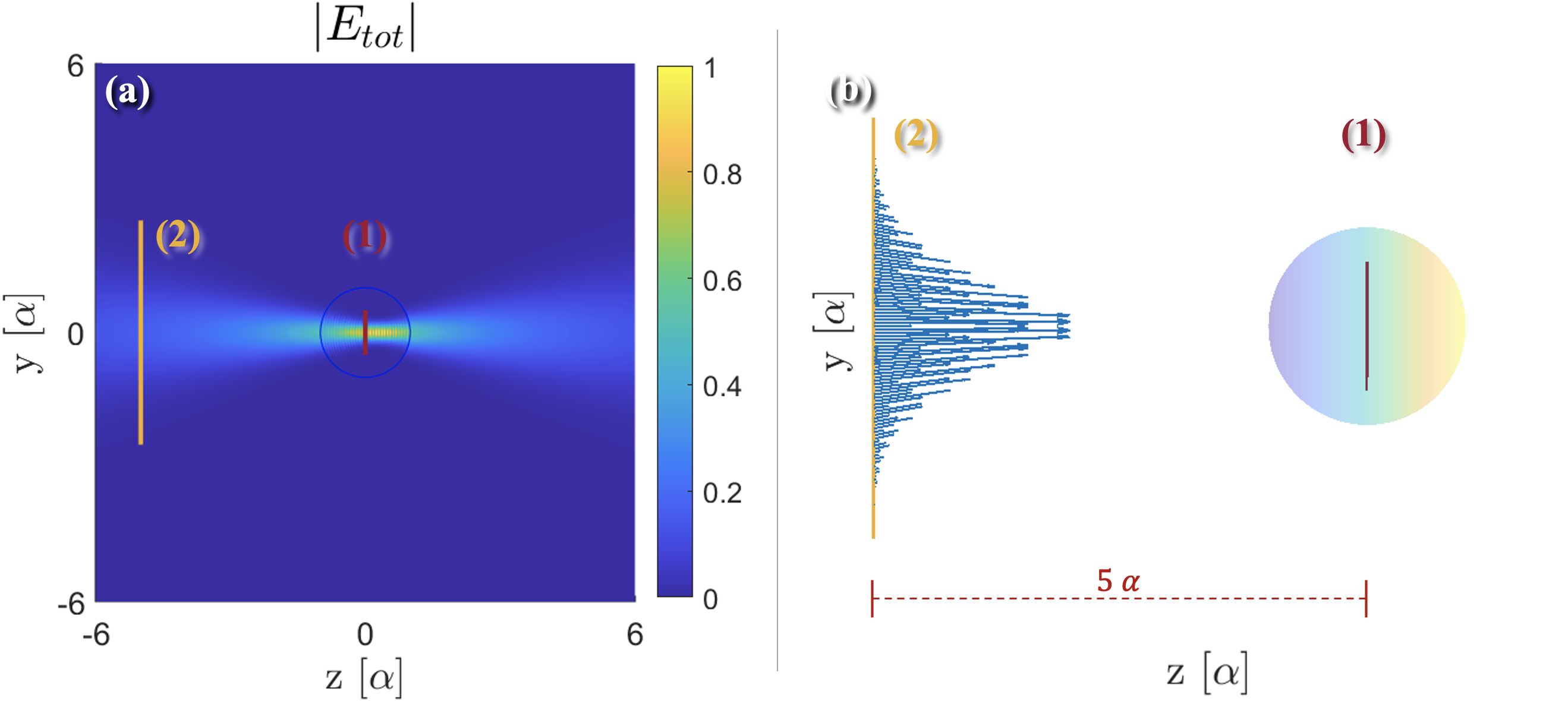}
\caption{\small Morphology-dependent resonance simulation of a sphere illuminated by Gaussian beam. On a) is a Gaussian beam synthesis on the source Plane 1 at the origin, and field components are computed on Plane 2. On b) is Simulation 2, where Plane 2 is used as a 3D ASM source, where each source point's local propagation direction is defined by the local Poynting vector (marked as blue), and source points are multiplied by complex amplitudes, obtained from Simulation 1.}
\label{GaussianMr136}
\end{figure}\FloatBarrier
The second simulation synthesizes the same Gaussian beam from Plane 2 with the 3D ASM using equal \(\lambda/6\) sampling space in the horizontal and vertical directions. Now each source point has locally varying propagation direction obtained as a Poynting vector \(\mathbf{P_{inc}}=0.5\Re\big\{\mathbf{E_{inc}}\times\mathbf{H^*_{inc}}\big\}\) from the incident field synthesized from the Plane 1. Also, each source point is multiplied by the local complex electric field \(\mathbf{E_{inc}}\) obtained from the first simulation, see simulation flow with source Planes 1 and 2 in Figure (\ref{SM}). \FloatBarrier
\begin{figure}[!ht]
\centering
\includegraphics[width=1\linewidth]{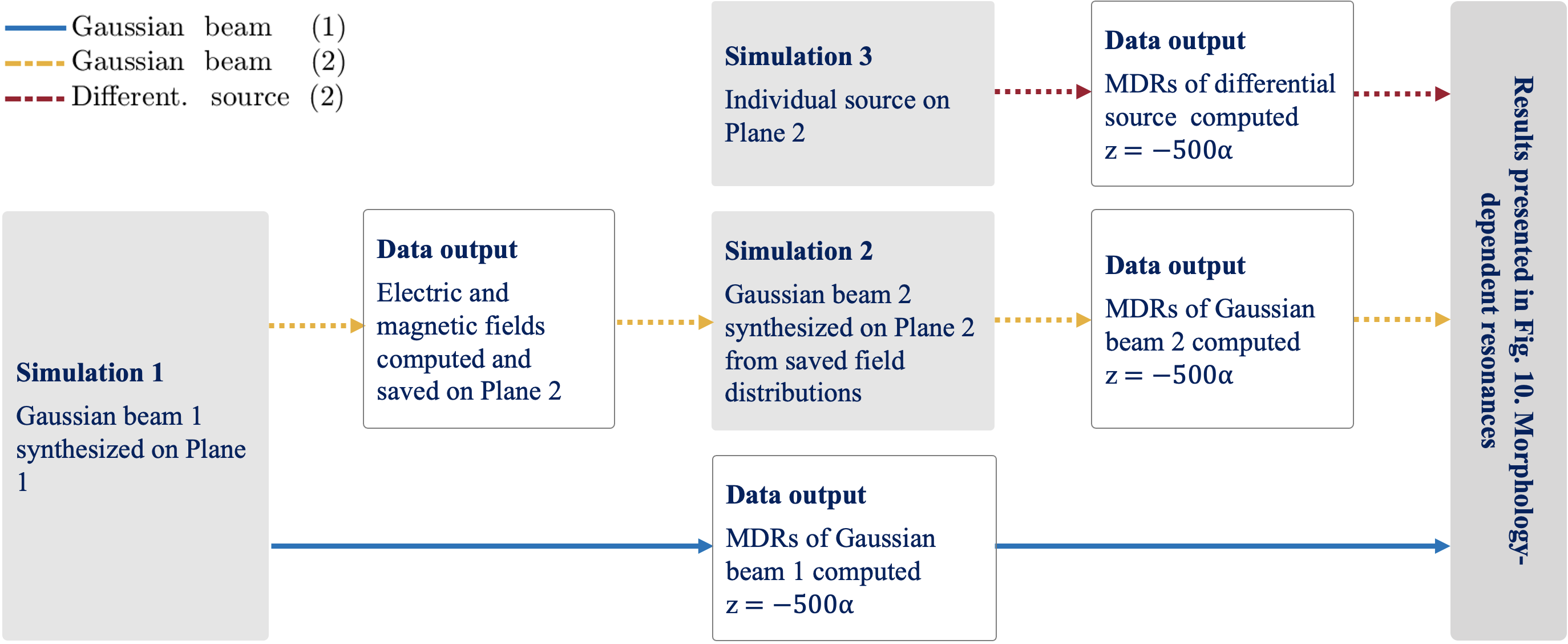}
\caption{\small Simulation flow graph. }
\label{SM}
\end{figure}
The Cartesian electric and magnetic field components of the incident field at Plane 2 are presented in Figure (\ref{GaussianCom}). These components are synthesized from Plane 1. Gaussian beam MDRs synthesized from both source plane scenarios are computed at -500 \(\alpha\) distance from the origin \cite{khaled} and compared in Figure (\ref{MDR}). The resonances are labeled as \(TE_{n,l}\) or \(TM_{n,l}\), where \(n\) indicates the mode number, and \(l\) the number of radial peaks in the angle-averaged internal energy density distribution \cite{khaled}.  
\begin{figure}[!h]
\centering
\includegraphics[width=1\linewidth]{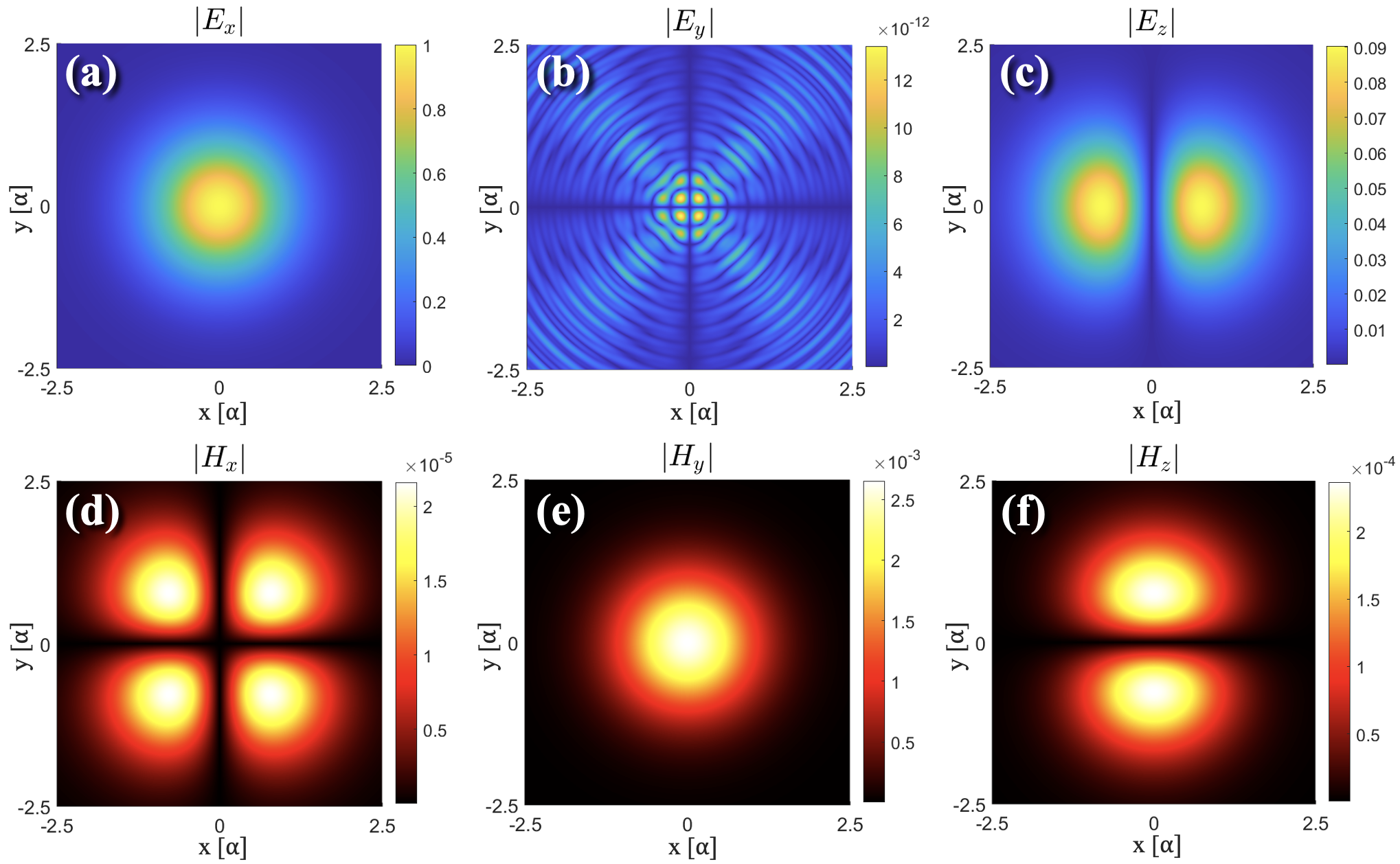}
\caption{\small Electric and magnetic field components of Gaussian beam one at the plane 2. On a)-c) are electric field components, and on d)-f) are magnetic field components. All the components are normalized by \(_{\max}|E_x|\).}
\label{GaussianCom}
\end{figure}\FloatBarrier
First, the homogeneous sphere is illuminated with a Gaussian beam focused at the origin. Second, a coated sphere is illuminated with a Gaussian beam focused on the side of the coated sphere \(y=\alpha\) to verify even more specific resonance behavior. Also, the MDRs from a source point on Plane 2 in both scenarios are presented to underline the resonance difference with varying beam shapes.
\begin{figure}[!h]
\centering
\includegraphics[width=1\linewidth]{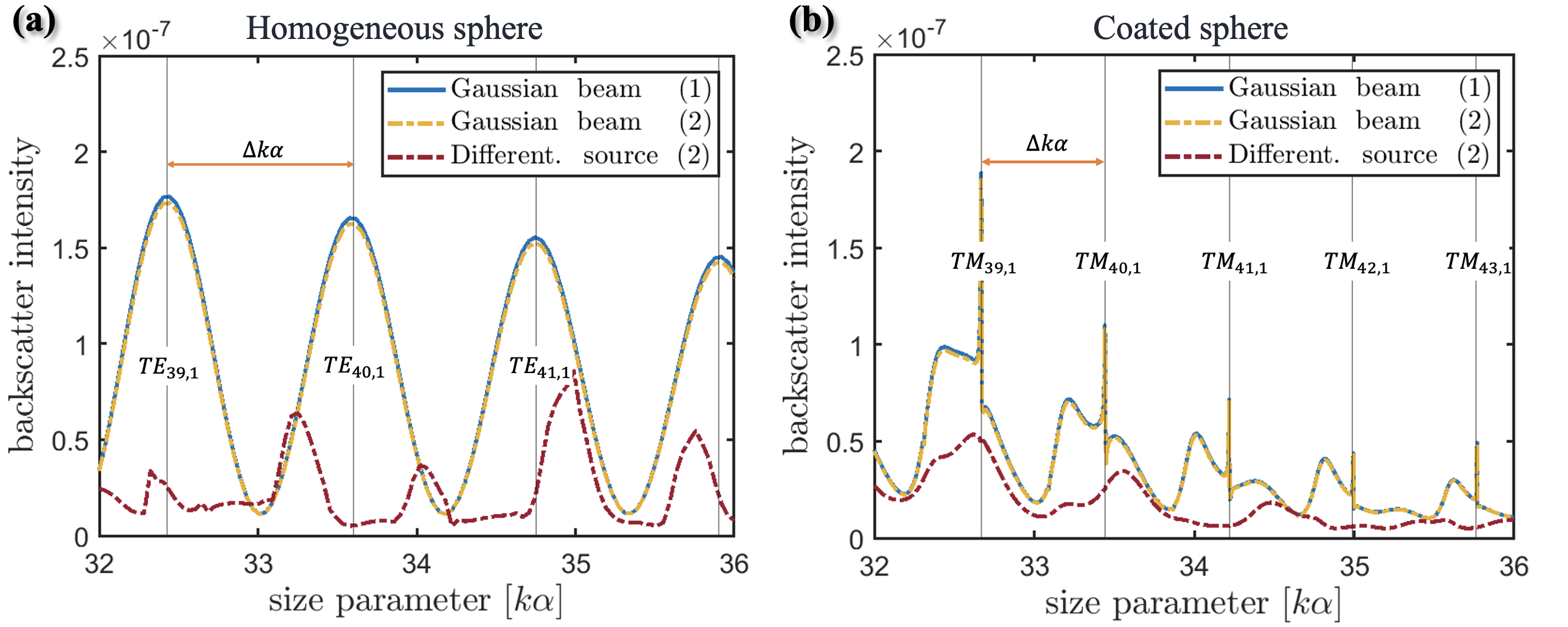}
\caption{\small Morphology-dependent resonances of a Gaussian beam and a differential source illumination presented in Figure (\ref{DiffSource}). On a) is a situation where the Gaussian beam waist is positioned at the origin of the homogeneous sphere with refractive index \(n_r=1.36\). On b)  is a situation where the Gaussian beam waist is on the side of the coated sphere, with a core refractive index of  \(n_r=1.5\) and shell as  \(n_r=1.36\). The backscatter intensities of one source point are scaled to fit the same plot.}
\label{MDR}
\end{figure}\FloatBarrier
Electromagnetic backscatter intensities show excellent MDRs matching from both source plane synthesis scenarios. The resonance shape, location, and peak spacing in Figure \ref{MDR} match with values from \cite{esam94}. In Figure (\ref{MDR} a), the peaks correspond to modes \(TE_{39,1},TE_{40,1}\), and \(TE_{41,1}\) respectively, and the spacing between the peaks satisfies the \(\Delta k\alpha=\pi/2m_r\approx1.15\) condition. Also, in Figure (\ref{MDR} b), the MDRs peaks from coated sphere correspond to the modes \(TM_{39,1},TM_{40,1},TM_{41,1}\) and \(TM_{42,1}\) respectively. 

The MDRs of an individual differential source do not match the Gaussian beam resonances. However, a total beam created as a superposition of the differential sources from the surface \(\Omega\) synthesizes the original beam with highly matching MDRs. This leads to the conclusion that the presented 3D AS method combined with the VSH expansion synthesizes electromagnetic scattering from multilayered spheres with high accuracy.\FloatBarrier

%%%%%%%%%%%%%%%%%%%%%%%%%%%%%%%%%%%%%%%%%%%%%%%%%%%%%
%%%%%%%%%%% Elliptical source distribution %%%%%%%%%%
%%%%%%%%%%%%%%%%%%%%%%%%%%%%%%%%%%%%%%%%%%%%%%%%%%%%%
\subsection{Elliptical source distribution}
An elliptical surface was used as an example of a source distribution to synthesize a focused linearly polarized incident field using \(\lambda/6\) sampling, see Figure (\ref{Ell}). The incident field is compared to the Physical optics simulations at the yz-plane with an excellent agreement to less than a -39 dB difference in amplitude and less than the 0.1-degree difference in phase. Also, the scattered fields are computed from a \(\alpha=\) 7.8 mm radius sphere in xz-plane with 100 dielectric layers with linearly varying permittivity \(\epsilon_r\in[1-i0.001,\ 3-i0.01]\) from the surface to the center.
\begin{figure}[!ht]
\centering
\includegraphics[width=0.55\linewidth]{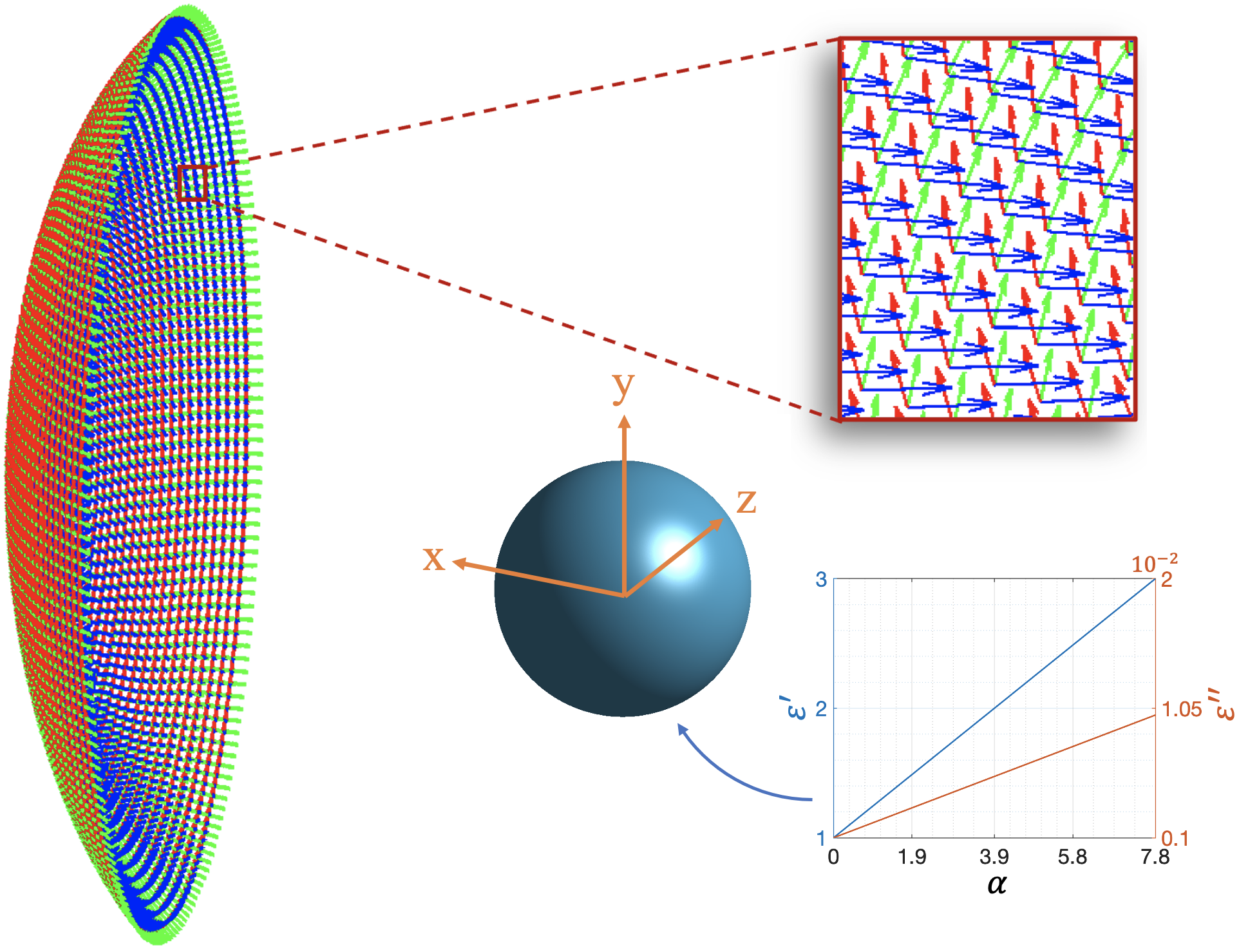}
\caption{\small Elliptical source distribution with 1424 source points. Local Cartesian base vectors (\(\mathbf{e}_1,\mathbf{e}_2,\mathbf{e}_3\)) are marked as red, green, and blue arrows, respectively.}
\label{Ell}
\end{figure}\FloatBarrier
Base vectors for elliptical surfaces are derived from the Eq. (\ref{eq1} - \ref{eq4}) with the parametrization
\begin{equation}\label{eqEllips}
\begin{array}{l}
\Omega = \Big\{\mathbf{o}(\theta,\phi)=(a_x\sin{\theta}\cos{\phi}-x_0)\mathbf{e}_x+(b_y\sin{\theta}\sin{\phi}-y_0)\mathbf{e}_y -(c_z\cos{\theta}+z_0)\mathbf{e}_z \ | \\ \theta\in[\pi/6,2\pi/3], \ \phi\in[-\pi/6,\pi/6]\Big\},
\end{array}
\end{equation}
with axial scaling factors \(a_x=0.6,\ b_y=1.5,\ c_z=0.8\) and an unscaled radius \(r=3\). The origin of the ellipsoid is shift \((x_0,y_0,z_0)=(3,0,-0.5)\) compared to the sphere. 

Figure (\ref{EllTot}) illustrate the amplitude and phase of the total field \(\mathbf{E}_{tot}=\mathbf{E}_{inc}+\mathbf{E}_{scat}\) on xz-plane where \(x,z\in[-5a,5a]\). Simulations are performed with \(500\times 500\) evaluation points with size parameter of \(k\alpha=13\) and VSH modes of \(N=50\).
\begin{figure}[!ht]
\centering
\includegraphics[width=1\linewidth]{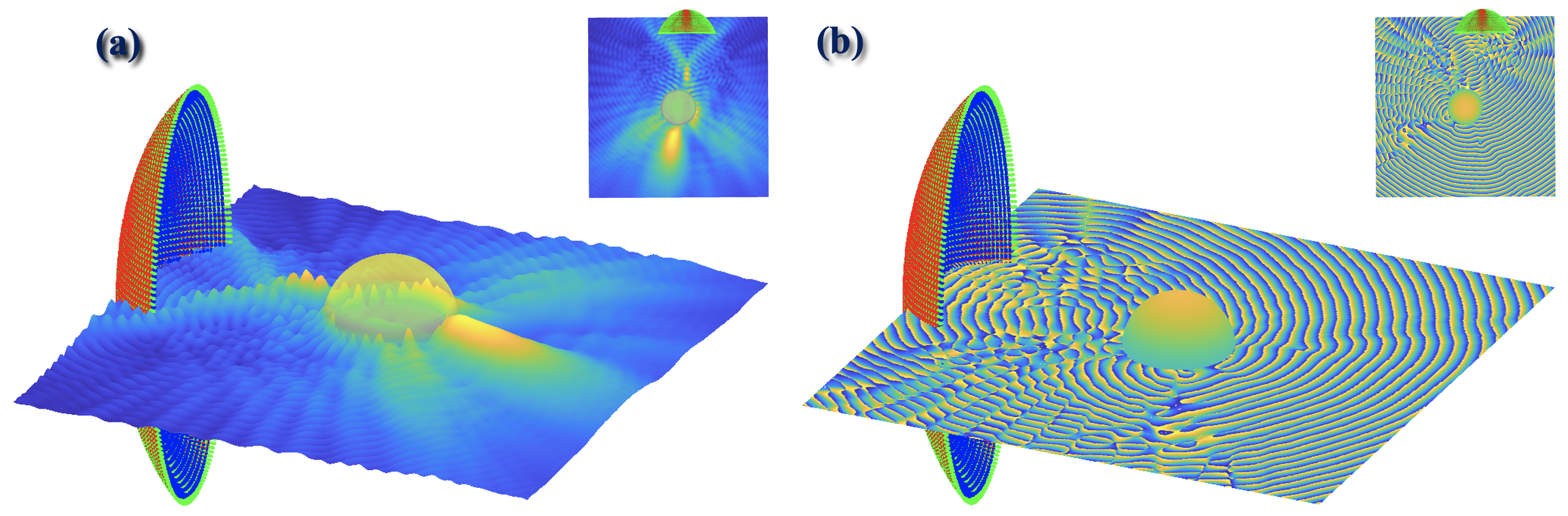}
\caption{\small Total fields from a 100-layer sphere illuminated with the incident beam defined in the ellipsoidal surface. a) total field magnitude, and b) total field phase.}
\label{EllTot}
\end{figure}\FloatBarrier
\section{Conclusions}
%%%%%%%%%%%%%%%%%%%%%%%%%%%%%%%%%%%%%%%%%%%%%%%%%%%%
We propose a method to synthesize electromagnetic beams by arbitrary surface electric field distributions. The goal is to compute scattering from multilayered spheres illuminated by incident fields that can be modified freely by adjusting the source distribution's shape, amplitude, phase, and location. The incident field is presented in VSH expansion defined by the BSCs, and the scattered field is obtained by mapping the incident field BSCs with the T-matrix method. 

The goal is achieved by computing BSCs with the proposed 3D ASM derived from the nominal 2D ASM in the theory chapter. The key concept is to divide the arbitrary surfaces into differential elements, which are used as local AS sources. Location and orientation information of AS sources is transformed in a global coordinate system and mapped into the BSCs to preserve spherical symmetry for VSH expansion. As a result, the total incident field is obtained as a superposition of VSH-expanded differential sources, in which orientation, locations, and complex amplitude can be adjusted. 

The proposed method approximates incident field synthesis from an arbitrary surface, considering the wavelength compared to the surface details. Then the scattered fields are rigorously computed from approximated incident fields with the Mie theory. In addition, two essential points of the method are (1) 3D ASM is analogous to the magnetic dipole presentation of the Stratton-Chu method, which satisfies Maxwell's equations and well approximates scattering/radiation from surfaces with RoC larger than \(\lambda/2\), and (2) when the 3D ASM is applied to planar surfaces, it returns the original 2D AS representation. 

Simulations verify the proposed method's incident field synthesis and scattering accuracy. First, we demonstrate the source points magnetic dipole behavior with the radiation patterns without singularities. Then we verify the scattering accuracy by comparing the MDRs resonances from spheres illuminated by the Gaussian beam created from the beam waist by nominal 2D ASM and by more complex electric field distribution with the proposed 3D ASM. The MDRs resonances of the Gaussian beam synthesized by the 3D ASM were in excellent agreement with the reference values, validating the high accuracy on the scattered fields. At last, we demonstrate the method's practicality by synthesizing an incident field from an elliptical surface and computing scattered fields from a 100-layer sphere. Additionally, we present mathematical proof to support the theory.

The novelty of this approach lies in a straightforward simulation algorithm, where the incident field can be defined at any parametrized surface. Furthermore, the parameterized surface can be located inside the computational domain, on the surface, or inside the spherical scatterer without restrictions and source singularities. Unrestricted placement of the source distribution has clear advantages, especially in the inverse beam design, where the desired beam is defined on the sub-region on the spherical scatterer. Additionally, the synthesized beam can be reradiated from any evaluation surface to another, enabling beam simulation between optical elements with the ability to consistently compute the scattered fields from multilayered spheres.
%%%%%%%%%%%%%%%%%%%%%%%%%%%%%%%%%%%%%%%%%%%%%%%%%%%%
\section{Appendix A}
%%%%%%%%%%%%%%%%%%%%%%%%%%%%%%%%%%%%%%%%%%%%%%%%%%%%

\begin{theorem}
Let the function \(E_0\) be continuous in a closed segment of surface \(\Omega\) and let \(\mathbf{r}\in\mathbb{R}^3\) be such that \((\mathbf{r}-\mathbf{o})\cdot(\frac{\partial\mathbf{o}}{\partial p}\times\frac{\partial\mathbf{o}}{\partial q})\neq0\) for all \(\mathbf{o}\in\Omega\). The field \(E\) created by function \(E_0\) is presented by Eq.(\ref{eq15}) and Eq.(\ref{eq16}) at point \(\mathbf{r}\). 
\end{theorem}

\begin{proof}
Let us define \(\epsilon>0\). Let us show first that when the partition {\(\Omega_t\)} is small enough, in other words, the areas \(|\Omega_t|\) are small enough regardless of t, by replacing the function \(E_0^t\) supported by each piece \(\Omega_t\) with a constant \(E_0^t(\bar{0},\bar{0})=E_0(p,q)\), the error made in Eq. (\ref{eq10}) is smaller than \(\epsilon/4\). Function \(E_0\) is uniformly continuous on the compact segment of surface \(\Omega\), and based on this, with sufficiently small \(|\Omega_t|\) it holds
 
\begin{equation}\label{eq33}
\begin{array}{l}
|E^t_0(\bar{x},\bar{y})-E_0(p,q)|\leq\frac{\epsilon}{4I|\Omega|}
\end{array}
\end{equation}
for all \(\bar{x},\bar{y}\in\Omega_t\) and \(t\), where the constant \(I\) will be defined later on. In the Fourier transform on the plane holds

\begin{equation}\label{eq34}
\begin{array}{l}
|\mathcal{F}\left\{E_0^t\right\}-\mathcal{F}\left\{E_0(p,q)\right\}|
\leq\iint_{\Omega_t}|E^t_0(\bar{x},\bar{y})-E_0(p,q)||e^{-i(k_{\bar{x}}\bar{x}+ k_{\bar{y}}\bar{y})}|d\bar{x}d\bar{y}\\\leq\iint_{\Omega_t}\frac{\epsilon}{4I\Omega}|e^{-i(k_{\bar{x}}\bar{x}+ k_{\bar{y}}\bar{y})}|d\bar{x}d\bar{y}
\leq\iint_{\Omega_t}\frac{\epsilon}{4I|\Omega|}d\bar{x}d\bar{y}=\frac{\epsilon|\Omega_t|}{4I|\Omega|},
\end{array}
\end{equation}
for all \((k_{\bar{x}},k_{\bar{y}})\in\mathbb{R}^2\). Function \(e^{i|\bar{z}|\sqrt{k^2-k_{\bar{x}}^2-k_{\bar{y}}^2}}\) is integrable on \((k_{\bar{x}},k_{\bar{y}})\) plane, because outside of the circle \(k_{\bar{x}}^2+k_{\bar{y}}^2=k^2\), the exponential becomes real and negative. Let \(I\) be, at first, a continuous elementary function

\begin{equation}\label{eq35}
\begin{array}{l}
I(|\bar{z}|)=\iint_{\mathbb{R}^2}|e^{i|\bar{z}|\sqrt{k^2-k_{\bar{x}}^2-k_{\bar{y}}^2}}|dk_{\bar{x}}dk_{\bar{y}},
\end{array}
\end{equation}
where \(\bar{z}=\bar{z}(\mathbf{o})=\bar{z}(p,q)\). Let r be as expected. Then because of Eq. (\ref{eq2}-\ref{eq3}, \ref{eq5}-\ref{eq7}) we have \((\mathbf{r}-\mathbf{o})\cdot\mathbf{e}_3(\mathbf{o})\neq0\) for all \(\mathbf{o}\in\Omega\). Thus, a positive, continuous function \(|\bar{z}(p,q)|\) gets its minimum value \(\bar{z}_{min}>0\) on the compact set \(\left\{(p,q)|p\in[p_1,p_2],q\in[q_1,q_2]\right\}\). Especially \(|\bar{z}|\geq \bar{z}_{min}\) for all \(t\) regardless of partitioning. Finally, the constant \(I\) is defined 

\begin{equation}\label{eq36}
\begin{array}{l}
I=I(z_{min})<\infty.
\end{array}
\end{equation}
Due the monotone of integral, \(I(|\bar{z}(p,q)|)\leq I\) on all \(t\) regardless of partitions. The upper limit for the error due \(E_0(p,q)\) on the Eq. (\ref{eq10}) is obtained as

\begin{equation}\label{eq37}
\begin{array}{l}
\frac{1}{4\pi^2}\iint_{\mathbb{R}^2}|\mathcal{F}\left\{E_0^t\right\}-\mathcal{F}\left\{E_0(p,q)\right\}|(k_{\bar{x}},k_{\bar{y}})|e^{i|\bar{z}|\sqrt{k^2-k_{\bar{x}}^2-k_{\bar{y}}^2}}|dk_{\bar{x}}dk_{\bar{y}}\\
\leq\frac{1}{4\pi^2}\frac{\epsilon|\Omega_t|}{4I|\Omega|}\iint_{\mathbb{R}^2}|e^{i|\bar{z}|\sqrt{k^2-k_{\bar{x}}^2-k_{\bar{y}}^2}}|dk_{\bar{x}}dk_{\bar{y}}\\\leq\frac{1}{4\pi^2}\frac{\epsilon|\Omega_t|}{4I|\Omega|}I\leq\frac{\epsilon|\Omega_t|}{4|\Omega|},
\end{array}
\end{equation}
where Eq. (\ref{eq33}) and Eq. (\ref{eq36}) have been used. Because of \(||\mathbf{e}_1||=1\) and the triangle inequality the error in Eq. (\ref{eq11}) is

\begin{equation}\label{eq38}
\begin{array}{l}
err_1\leq\sum_t\frac{\epsilon|\Omega_t|}{4|\Omega|}=\frac{\epsilon}{4}.
\end{array}
\end{equation}
In the above approximation, a constant \(E_0(p,q)\) is used on piece \(\Omega\); as a function, it is zero on the plane outside of the piece. Let us denote this function simply as \(E_0(p,q)\). Next, we consider the Fourier transform of the function \(E_0(p,q)\). Because function \(e^{i|\bar{z}|\sqrt{k^2-k_{\bar{x}}^2-k_{\bar{y}}^2}}\) is integrable, there can be found an origin centered closed disk \(\bar{B}\) such that
\begin{equation}\label{eq39}
\begin{array}{l}
\left[max_\Omega|E_0|\right]\iint_{\mathbb{R}^2/\bar{B}}|e^{i|\bar{z}|\sqrt{k^2-k_{\bar{x}}^2-k_{\bar{y}}^2}}|dk_{\bar{x}}dk_{\bar{y}}<\frac{\epsilon}{4|\Omega|}
\end{array}
\end{equation}
for all \(t\) because \(|\bar{z}|\geq \bar{z}_{min}\). For all \((k_{\bar{x}},k_{\bar{y}})\in\mathbb{R}^2\) it holds

\begin{equation}\label{eq40}
\begin{array}{l}
|\mathcal{F}\left\{E_0(p,q)\right\}(k_{\bar{x}},k_{\bar{y}})|\\\leq \left[max_\Omega|E_0|\right]\iint_{\Omega_t}1dk_{\bar{x}}dk_{\bar{y}}=max|E_0||\Omega_t|.
\end{array}
\end{equation}
Let us now consider the error (\(err_2\)) made in integral Eq. (\ref{eq10}) when \(\mathcal{F}\left\{E_0(p,q)\right\}\) is replaced by the constant \(E_0(p_t,q_t)|\Omega_t|\). Under the estimates Eq. (\ref{eq39}) and Eq. (\ref{eq40}) it follows

\begin{equation}\label{eq41}
\begin{array}{l}
\iint_{\mathbb{R}^2/\bar{B}}|\mathcal{F}\left\{E_0(p,q)\right\}-E_0(\theta,\phi)|\Omega_t|||e^{i|\bar{z}|\sqrt{k^2-k_{\bar{x}}^2-k_{\bar{y}}^2}}|dk_{\bar{x}}dk_{\bar{y}}\\
\leq\iint_{\mathbb{R}^2/\bar{B}}2\left[max_\Omega|E_0|\right]|\Omega_t||e^{i|\bar{z}|\sqrt{k^2-k_{\bar{x}}^2-k_{\bar{y}}^2}}|dk_{\bar{x}}dk_{\bar{y}}\\
\leq 2|\Omega_t|\frac{\epsilon}{4|\Omega|}=\frac{\epsilon|\Omega_t|}{2|\Omega|}.
\end{array}
\end{equation}

For all \((k_{\bar{x}},k_{\bar{y}})\in\bar{B}\) and (\(\bar{x},\bar{y})\in\Omega_t\), we can estimate by continuity of the function

\begin{equation}\label{eq42}
\begin{array}{l}
\left[max_\Omega|E_0|\right]|e^{-i(k_{\bar{x}}\bar{x}+k_{\bar{y}}\bar{y})}-1|\leq\frac{\epsilon}{4|\bar{B}|\Omega|}
\end{array}
\end{equation}
always when \(|\Omega_t|\) is small enough, because then \(\bar{x}\approx0, \bar{y}\approx0\) and hence \(k_{\bar{x}}\bar{x}+k_{\bar{y}}\bar{y}\approx0\ -\) that is the assignment of \(\bar{B}\). Then in the Fourier transform at point \((k_{\bar{x}},k_{\bar{y}})\in\bar{B}\) it holds

\begin{equation}\label{eq43}
\begin{array}{l}
|\mathcal{F}\left\{E_0(p,q)\right\}(k_{\bar{x}},k_{\bar{y}})-E_0(p,q)|\Omega_t||\\
=|\iint_{\Omega_t}E_0(p,q)e^{-i(k_{\bar{x}}\bar{x}+k_{\bar{y}}\bar{y})}d\bar{x}d\bar{y}-E_0(p,q)1d\bar{x}d\bar{y}|\\
\leq|E_0(p,q)|\iint_{\Omega_t}|e^{-i(k_{\bar{x}}\bar{x}+k_{\bar{y}}\bar{y})}-1|d\bar{x}d\bar{y}\\
\leq\iint_{\Omega_t}\frac{\epsilon}{4|\bar{B}||\Omega|}d\bar{x}d\bar{y}=\frac{\epsilon|\Omega_t|}{4|\bar{B}||\Omega|}.
\end{array}
\end{equation}
Replacing by \(\mathcal{F}\left\{E_0(p,q)\right\}\) with the constant \(E_0(p,q)|\Omega_t|\) in Eq. (\ref{eq10}) we make an error

\begin{equation}\label{eq44}
\begin{array}{l}
\frac{1}{4\pi^2}\Big(\iint_{\bar{B}}+\iint_{\mathbb{R}^2/\bar{B}}\Big)|\mathcal{F}\left\{E_0(p,q)\right\}-E_0(p,q)|\Omega_t||e^{i|\bar{z}|\sqrt{k^2-k_{\bar{x}}^2-k_{\bar{y}}^2}}|dk_{\bar{x}}dk_{\bar{y}}\\
\leq\frac{\epsilon|\Omega_t|}{4|\bar{B}||\Omega_t|}\iint_{\bar{B}}1dk_{\bar{x}}dk_{\bar{y}}+\frac{\epsilon|\Omega_t|}{2|\Omega|}=\frac{\epsilon|\Omega_t|}{4|\Omega|}+\frac{\epsilon|\Omega_t|}{2|\Omega|}=\frac{3\epsilon|\Omega_t|}{4|\Omega|},
\end{array}
\end{equation}
where the estimates Eq. (\ref{eq41}) and Eq. (\ref{eq43}) have been used. The error for sum Eq. (\ref{eq11}) is obtained as

\begin{equation}\label{eq45}
\begin{array}{l}
err_{2}\leq\sum_t\frac{3\epsilon|\Omega_t|}{4|\Omega|}=\frac{3\epsilon}{4}.
\end{array}
\end{equation}
Thus, when the partition {\(\Omega_t\)} is dense enough and \(\mathbf{r}\in\mathbb{R}^3\) be such that \((\mathbf{r}-\mathbf{o})\cdot(\frac{\partial\mathbf{o}}{\partial p}\times\frac{\partial\mathbf{o}}{\partial q})\neq0\) for all \(\mathbf{o}\in\Omega\), we get by Eq. (\ref{eq38}) and Eq. (\ref{eq45})

\begin{equation}\label{eq46}
\begin{array}{l}
\Big|\Big|\sum_tE_t(\mathbf{r})\mathbf{e}_1-\frac{1}{4\pi^2}\sum_tE_0(p,q)\mathbf{e}_1(p,q)|\Omega_t|\iint_\mathbb{R}e^{i(k_{\bar{x}}\bar{x}+ k_{\bar{y}}\bar{y})}e^{i|\bar{z}|\sqrt{k^2-k_{\bar{x}}^2-k_{\bar{y}}^2}}dk_{\bar{x}}dk_{\bar{y}}\Big|\Big|\\
\leq err_1+err_{2}
\leq\frac{\epsilon}{4}+\frac{3\epsilon}{4}=\epsilon.
\end{array}
\end{equation}
Finally, by continuity of the integrand, the integral in Eq. (\ref{eq15}) exists, and, because of the estimate Eq. (\ref{eq46}), for the points \(\mathbf{r}\in\mathbb{R}^3\) like in Theorem there also exists

\begin{equation}\label{eq47}
\begin{array}{l}
\mathbf{E}_1(\mathbf{r})=\mathbf{E}_1(x,y,z)=lim_{|\Omega_t|\to 0}\sum_tE_t(\mathbf{r})\mathbf{e}_1\\
=\frac{1}{4\pi^2}\iint_{\Omega}E_0(p,q)\mathbf{E}_t(\mathbf{r};p,q)\|\frac{\partial o}{\partial p}\times\frac{\partial o}{\partial q}\|dpdq.
\end{array}
\end{equation}

\end{proof}

The proof above also holds for the second Eq. (\ref{eq19}) component; only a few adaptions are needed. Moreover, Theorem 5.1 holds with an expectation \(\mathbf{r}\notin\Omega\) for both Eq. (\ref{eq19}) components, and the solutions are visible in the simulations. The required extended proof is reasonably complicated and will be presented in a separate article. Its difficulty lies primarily in that integrals do not exist in the usual sense when \(\bar{z}=0\). 
%%%%%%%%%%%%%%%%%%%%%%%%%%%%%%%%%%%%%%%%%%%%%%%%%%%%
\section{Appendix B}
%%%%%%%%%%%%%%%%%%%%%%%%%%%%%%%%%%%%%%%%%%%%%%%%%%%%

The VSH beam shape  coefficients for each \(ij-\) plane wave of the incident field is defined as \cite{barber}
\begin{equation}\label{eqB1}
\begin{array}{l}
a^t_{{\text{emn}}_{ij}}=4i^n\mathbf{e}_{o_{ij}}\cdot\big[-\mathbf{e}_{\theta}\sin{(m\phi)}\frac{m}{\sin{\theta}}P^m_n(\cos{\theta})-\mathbf{e}_{\phi}\cos{(m\phi)}\frac{d}{d\theta}P^m_n(\cos{\theta})\big],\\
a^t_{{\text{omn}}_{ij}}=4i^n\mathbf{e}_{o_{ij}}\cdot\big[\mathbf{e}_{\theta}\cos{(m\phi)}\frac{m}{\sin{\theta}}P^m_n(\cos{\theta})-\mathbf{e}_{\phi}\sin{(m\phi)}\frac{d}{d\theta}P^m_n(\cos{\theta})\big],\\
b^t_{{\text{emn}}_{ij}}=-4i^{n+1}\mathbf{e}_{o_{ij}}\cdot\big[\mathbf{e}_{\theta}\cos{(m\phi)}\frac{m}{\sin{\theta}}P^m_n(\cos{\theta})-\mathbf{e}_{\phi}\sin{(m\phi)}\frac{d}{d\theta}P^m_n(\cos{\theta})\big],\\
b^t_{{\text{omn}}_{ij}}=-4i^{n+1}\mathbf{e}_{o_{ij}}\cdot\big[\mathbf{e}_{\theta}\sin{(m\phi)}\frac{m}{\sin{\theta}}P^m_n(\cos{\theta})+\mathbf{e}_{\phi}\cos{(m\phi)}\frac{d}{d\theta}P^m_n(\cos{\theta})\big],
\end{array}
\end{equation}
where \(P^m_n\) is the associated Legendre function of the first kind of degree \(n\) and order \(m\). Additionally \(\theta,\phi\) are spherical coordinates and \(\mathbf{e}_\theta,\mathbf{e}_\phi\) spherical base vectors of \(\mathbf{r}\) in the local base \((\mathbf{e}_1,\mathbf{e}_2,\mathbf{e}_3)\) as in Eq. (\ref{eq27})
The scattered field from a differential element in VSH presentation is obtained as
\begin{equation}\label{eqB2}
\begin{array}{l}
\begin{aligned}
\mathbf{E}_{scat}(\mathbf{r})&=p^2\sum_m\sum_n D_{mn}\Big[f_{\text{emn}}\mathbf{M}^{(3)}_{\text{emn}}(k\mathbf{r})+f_{\text{omn}}\mathbf{M}^{(3)}_{\text{omn}}(k\mathbf{r})\\&+g_{\text{emn}}\mathbf{N}^{(3)}_{\text{emn}}(k\mathbf{r})+g_{\text{omn}}\mathbf{N}^{(3)}_{\text{omn}}(k\mathbf{r})\Big],
\end{aligned}
\end{array}
\end{equation}
where the superscripts (3) present the vector spherical harmonics (outgoing wave) with spherical Hankel function of the first kind \(h_{\text{n}}^{(1)}(kr)\). The \(f_{\text{emn}},f_{\text{omn}},g_{\text{emn}}\) and \(g_{\text{omn}}\) are vector spherical harmonic coefficient for the scattered field calculated from the T-matrix method as
\begin{equation}\label{eqB3}
\begin{array}{l}
[f_{\text{emn}} \ f_{omn} \ g_{\text{emn}} \ g_{omn}]^T=T[a_{\text{emn}} \ a_{omn} \ b_{\text{emn}} \ b_{omn}]^T
\end{array}
\end{equation}
where the T-matrix elements for the multilayered sphere are obtained by the algorithm defined in \cite{pena}.

%%%%%%%%%%%%%%%%%%%%%%% References 
%%%%%%%%%% If using BibTeX:
%\printbibliography %Prints bibliography
\bibliography{sample}
%%%%%%%%%% If preparing manually:
% \newcommand{\enquote}[1]{``#1''}

% \bibitem{Zhang:14}
% Y.~Zhang, S.~Qiao, L.~Sun, Q.~W. Shi, W.~Huang, L.~Li, and Z.~Yang,
%   \enquote{Photoinduced active terahertz metamaterials with nanostructured
%   vanadium dioxide film deposited by sol-gel method,}
%   {\protect\JournalTitle{Optics Express}} \textbf{22}, 11070--11078 (2014).

% {Optical Society}, \enquote{{OSA Publishing},}
%   \url{http://www.osapublishing.org}.

% \bibitem{FORSTER2007}
% P.~Forster, V.~Ramaswamy, P.~Artaxo, T.~Bernsten, R.~Betts, D.~Fahey,
%   J.~Haywood, J.~Lean, D.~Lowe, G.~Myhre, J.~Nganga, R.~Prinn, G.~Raga,
%   M.~Schulz, and R.~V. Dorland, \enquote{Changes in atmospheric consituents and
%   in radiative forcing,} in \enquote{Climate Change 2007: The Physical Science
%   Basis. Contribution of Working Group 1 to the Fourth assesment report of
%   Intergovernmental Panel on Climate Change,}  S.~Solomon, D.~Qin, M.~Manning,
%   Z.~Chen, M.~Marquis, K.~B. Averyt, M.~Tignor, and H.~L. Miler, eds.
%   (Cambridge University Press, 2007).

% \end{thebibliography}

\end{document}